\documentclass[3p]{elsarticle}

\usepackage{amssymb}
\usepackage{amsthm}
\usepackage{upgreek}
\usepackage{enumitem}
\usepackage{amsmath,amsthm,a4wide,mathrsfs,color,array}
\usepackage[utf8]{inputenc}
\usepackage[T1]{fontenc}
\usepackage{subcaption}
\usepackage{upgreek}
\usepackage{graphicx}
\usepackage{float}
\usepackage{multirow}
\usepackage{hyperref}
\usepackage{cleveref}
\usepackage{tikz-cd}
\usepackage{fancyvrb}
\usepackage{xcolor}
\usepackage{hhline}
\usetikzlibrary{shapes,arrows,positioning}
\usepackage{lmodern,bm} 

\DeclareMathAlphabet\mathbfcal{OMS}{cmsy}{b}{n}
\usepackage{tikz}
\usetikzlibrary{shapes,arrows,positioning}
\usetikzlibrary{babel}
\usetikzlibrary{shapes,arrows}
\definecolor{dark_blue_pers}{RGB}{46,87,144}
\definecolor{blue_pers}{RGB}{54,104,171}
\definecolor{grey_pers}{RGB}{245,245,245}
\definecolor{red_pers}{RGB}{213,78,33}
\usetikzlibrary{backgrounds, scopes}

\newtheorem{problem}{Problem}
\newtheorem{remark}{Remark}

\newtheorem{proposition}{Proposition}

\usepackage{color}
\definecolor{dark_blue}{RGB}{46,87,144}
\definecolor{dark_green}{RGB}{0,100,0}

\newcommand{\nc}{\newcommand}

\nc{\norm}[2]{\left\|#1\right\|_{#2}}

\usepackage{amsmath}
\usepackage{algpseudocode}
\usepackage{algorithm}

\algnewcommand\algorithmicforeach{\textbf{for each:}}
\algnewcommand\ForEach{\item[ \algorithmicforeach]}

\nc{\IC}{\mathbb{C}}
\nc{\IE}{\mathbb{E}}
\nc{\IN}{\mathbb{N}}
\nc{\IR}{\mathbb{R}}
\nc{\IH}{\mathbb{H}}
\nc{\IX}{\mathbb{X}}

\nc{\be}{\begin{equation}}
\nc{\ee}{\end{equation}}

\nc{\mB}{\mathcal{B}}
\nc{\mE}{\mathcal{E}}
\nc{\mI}{\mathcal{I}}
\nc{\mK}{\mathcal{K}}
\nc{\mN}{\mathcal{N}}
\nc{\mL}{\mathcal{L}}
\nc{\mT}{\mathcal{T}}
\nc{\mO}{\mathcal{O}}
\nc{\mH}{\mathcal{H}}
\nc{\mS}{\mathcal{S}}
\renewcommand{\d}{\operatorname{d}\!}
\nc{\tmI}{\widetilde{\mathcal{I}}}
\nc{\bxi}{\bm{\xi}}
\DeclareMathOperator{\sign}{sign}

\nc{\dofs}{\texttt{dofs}}

\DeclareMathAlphabet\bfcal{OMS}{cmsy}{b}{n}

\newcommand{\bmI}{\boldsymbol{\mathcal{I}}}

\DeclareMathAlphabet\mathbfcal{OMS}{cmsy}{b}{n}

\makeatletter
\newsavebox{\@brx}
\newcommand{\llangle}[1][]{\savebox{\@brx}{\(\m@th{#1\langle}\)}%
  \mathopen{\copy\@brx\kern-0.5\wd\@brx\usebox{\@brx}}}
\newcommand{\rrangle}[1][]{\savebox{\@brx}{\(\m@th{#1\rangle}\)}%
  \mathclose{\copy\@brx\kern-0.5\wd\@brx\usebox{\@brx}}}
\makeatother

\newcommand{\bmP}{\boldsymbol{\mathcal{P}}}
\newcommand{\bmA}{\boldsymbol{\mathcal{A}}}
\newcommand{\hbmA}{\mathbfcal{\hat{A}}}
\newcommand{\tbmA}{\mathbfcal{\tilde{A}}}
\newcommand{\bmM}{\boldsymbol{\mathcal{M}}}

\newcommand{\eps}{\varepsilon}

\newcommand{\divg}{\operatorname{div}}

\newcommand{\bcurl}{\operatorname{\bf{curl}}}
\newcommand{\bgamma}{{\boldsymbol{\gamma}}}

\newcommand{\bmX}{\bm{\mathcal{X}}}

\nc{\bx}{{\bf x}}
\nc{\by}{{\bf y}}
\nc{\bn}{{\bf n}}
\nc{\bq}{{\bf q}}
\nc{\bb}{{\bf b}}
\nc{\bz}{{\bf z}}
\nc{\brr}{{\bf r}}
\nc{\bu}{{\bf u}}
\nc{\bd}{{\bf d}}
\nc{\bv}{{\bf v}}
\nc{\bp}{{\bf p}}
\nc{\bff}{{\bf f}}

\nc{\bM}{{\bf M}}
\nc{\bW}{{\bf W}}
\nc{\bE}{{\bf E}}
\nc{\bV}{{\bf V}}
\nc{\bH}{{\bf H}}
\nc{\bL}{{\bf L}}
\nc{\bI}{{\bf I}}
\nc{\bP}{{\bf P}}

\newcommand{\bsH}{\boldsymbol{H}}
\newcommand{\bsL}{\boldsymbol{L}}
\newcommand{\bpsi}{\boldsymbol{\psi}}
\newcommand{\bphi}{\boldsymbol{\phi}}
\nc{\bU}{{\bf U}}

\nc{\btheta}{{\boldsymbol{\theta}}}

\DeclareFontFamily{U}{mathx}{\hyphenchar\font45}
\DeclareFontShape{U}{mathx}{m}{n}{
      <5> <6> <7> <8> <9> <10>
      <10.95> <12> <14.4> <17.28> <20.74> <24.88>
      mathx10
      }{}
\DeclareSymbolFont{mathx}{U}{mathx}{m}{n}
\DeclareFontSubstitution{U}{mathx}{m}{n}
\DeclareMathSymbol{\bigtimes}{1}{mathx}{"91}

\nc{\loc}{{_\textup{loc}}}

\usepackage[normalem]{ulem}

\journal{Journal of Computational Physics}

\begin{document}

\begin{frontmatter}

\title{Local Multiple Traces Formulation for Heterogeneous Electromagnetic Scattering: Implementation and Preconditioning}

\author[inria]{Paul Escapil-Inchausp\'e}
\ead{paul.escapil@inria.cl}
\author[BATH]{Carlos Jerez-Hanckes\corref{cor1}}
\ead{cjh239@bath.ac.uk}

\address[inria]{Inria Chile, Santiago, Chile}
\address[BATH]{Department of Mathematical Sciences, University of Bath, UK}
\cortext[cor1]{Corresponding author}

\begin{abstract}We consider the three-dimensional time-harmonic electromagnetic (EM) wave scattering transmission problem involving heterogeneous scatterers. The fields are approximated using the local multiple traces formulation (MTF), originally introduced in \cite{hiptmairjerez2012multiplehelmholtz} for acoustic scattering. This scheme assigns independent boundary unknowns to each subdomain and weakly enforces Calderón identities along with interface transmission conditions. As a result, the MTF effectively handles shared points or edges among multiple subdomains, while supporting various preconditioning and parallelization strategies. Nevertheless, implementing standard solvers presents significant challenges, particularly in managing the degrees of freedom associated with subdomains and their interfaces. To address these difficulties, we propose a novel framework that suitably defines approximation spaces and enables the efficient exchange of normal vectors across subdomain boundaries. This framework leverages the skeleton mesh, representing the union of all interfaces, as the computational backbone, and constitutes the first scalable implementation of the EM MTF. Furthermore, we conduct several numerical experiments, exploring the effects of increasing subdomains and block On-Surface-Raditation-Condition (OSRC) preconditioning, to validate our approach and provide insights for future developments.
\end{abstract}


\end{frontmatter}
\section{Introduction}\label{sec:intro}
We address the solution of electromagnetic (EM) wave transmission problems involving heterogeneous, open, and bounded scatterers $\Omega_S\in\mathbb{R}^3$. Specifically, we consider objects such that, for $M\in\mathbb{N}$, one can write 
$$\overline{\Omega}_S:=\bigcup_{i=1}^M\overline{\Omega}_i, \quad\Omega_0:=\mathbb{R}^3\setminus\overline{\Omega}_S,\quad \Gamma_i:=\partial\Omega_i,$$
with interfaces $\Gamma_{ij}:=\Gamma_i\cap\Gamma_j$ (cf.~Figure \ref{fig:Skeleton}), and for which physical parameters values inside each subdomain $\Omega_i$ vary. Such problems remain a significant challenge within the EM community, particularly, in view of the ever-wider use of the EM spectrum. Among various modeling approaches, boundary integral equations (BIEs) and their different discretization methods---such as boundary elements, spectral, Nystr\"om, collocation---are regularly employed to simulate EM fields, largely due to their ability to handle unbounded media effectively \cite{Nedelec,buffa2003boundary,buffa2003galerkin}, with many available commercial \cite{martyanov2014ansys,solvers2020cst}, and open-source  \cite{betcke2021bempp,DOLZ2020100476,krcools_2024_11174547,alouges2018fem} solvers. 

The Poggio-Miller-Chang-Harrington-Wu-Tsai (PMCHWT) formulation \cite{POM73,PMCHWT} is an ubiquitous choice for modeling EM transmission scattering problems involving multiple, possibly non-disjoint, homogeneous subdomains. In this formulation, a \emph{single pair} of electric and magnetic currents is assigned per interface $\Gamma_{ij}$ and, consequently, these currents must be reoriented when computing couplings between interfaces. Although the resulting PMCHWT matrices are inherently poorly conditioned, the formulation supports the use of multiplicative (operator) Calderón preconditioning \cite{CalderonPMCHWT,buffa2005remarks,claeys2012multi} and accelerated versions \cite{kleanthous2019calderon,KLEANTHOUS2022111099,ESCAPILINCHAUSPE2021220,escapil2019fast} for \emph{disjoint scatterers}. However, a significant limitation of the PMCHWT formulation is the lack of clear preconditioning strategies in cases involving several adjacent subdomains\footnote{Throughout, we will refer to such scatterers as \emph{complex}.}---i.e.,~when an edge or a vertex is shared by more than two subdomains (cf.~Figure \ref{fig:Skeleton}). The need for parallel and preconditioning-ready numerical schemes to handle wave scattering complex objects has driven substantial research in this area over the past decade.

\tikzset{every picture/.style={line width=0.75pt}}
\begin{figure}[t]
\centering
\tikzset{every picture/.style={line width=0.75pt}} 

\begin{tikzpicture}[x=0.75pt,y=0.75pt,yscale=-1,xscale=1]

\draw    (332.82,34) -- (265.71,86.73) ;
\draw    (332.82,34) -- (391.95,86.73) ;
\draw    (292.88,163.43) -- (273.7,208.17) ;
\draw    (292.88,163.43) -- (265.71,86.73) ;
\draw    (382.36,166.62) -- (391.95,86.73) ;
\draw    (273.7,208.17) -- (356.79,233.73) ;
\draw [color={rgb, 255:red, 46; green, 77; blue, 152 }  ,draw opacity=1 ]   (292.88,163.43) -- (382.36,166.62) ;
\draw    (356.79,233.73) -- (382.36,166.62) ;
\draw    (382.36,166.62) -- (436.69,190.59) ;
\draw [color={rgb, 255:red, 46; green, 77; blue, 152 }  ,draw opacity=1 ]   (391.95,86.73) -- (436.69,190.59) ;
\draw    (356.79,233.73) -- (436.69,190.59) ;
\draw [color={rgb, 255:red, 65; green, 117; blue, 5 }  ,draw opacity=1 ]   (337.62,165.02) -- (336.34,181.87) ;
\draw [shift={(336.11,184.86)}, rotate = 274.34] [fill={rgb, 255:red, 65; green, 117; blue, 5 }  ,fill opacity=1 ][line width=0.08]  [draw opacity=0] (6.25,-3) -- (0,0) -- (6.25,3) -- cycle    ;
\draw  [draw opacity=0][fill={rgb, 255:red, 248; green, 82; blue, 28 }  ,fill opacity=1 ] (289.14,163.43) .. controls (289.14,161.36) and (290.81,159.68) .. (292.88,159.68) .. controls (294.94,159.68) and (296.62,161.36) .. (296.62,163.43) .. controls (296.62,165.49) and (294.94,167.17) .. (292.88,167.17) .. controls (290.81,167.17) and (289.14,165.49) .. (289.14,163.43) -- cycle ;
\draw  [draw opacity=0][fill={rgb, 255:red, 248; green, 82; blue, 28 }  ,fill opacity=1 ] (353.05,233.73) .. controls (353.05,231.67) and (354.73,229.99) .. (356.79,229.99) .. controls (358.86,229.99) and (360.53,231.67) .. (360.53,233.73) .. controls (360.53,235.8) and (358.86,237.48) .. (356.79,237.48) .. controls (354.73,237.48) and (353.05,235.8) .. (353.05,233.73) -- cycle ;
\draw  [draw opacity=0][fill={rgb, 255:red, 248; green, 82; blue, 28 }  ,fill opacity=1 ] (378.62,166.62) .. controls (378.62,164.56) and (380.29,162.88) .. (382.36,162.88) .. controls (384.43,162.88) and (386.1,164.56) .. (386.1,166.62) .. controls (386.1,168.69) and (384.43,170.36) .. (382.36,170.36) .. controls (380.29,170.36) and (378.62,168.69) .. (378.62,166.62) -- cycle ;
\draw  [draw opacity=0][fill={rgb, 255:red, 248; green, 82; blue, 28 }  ,fill opacity=1 ] (388.2,86.73) .. controls (388.2,84.66) and (389.88,82.99) .. (391.95,82.99) .. controls (394.01,82.99) and (395.69,84.66) .. (395.69,86.73) .. controls (395.69,88.8) and (394.01,90.47) .. (391.95,90.47) .. controls (389.88,90.47) and (388.2,88.8) .. (388.2,86.73) -- cycle ;
\draw  [draw opacity=0][fill={rgb, 255:red, 248; green, 82; blue, 28 }  ,fill opacity=1 ] (432.94,190.59) .. controls (432.94,188.52) and (434.62,186.85) .. (436.69,186.85) .. controls (438.75,186.85) and (440.43,188.52) .. (440.43,190.59) .. controls (440.43,192.66) and (438.75,194.33) .. (436.69,194.33) .. controls (434.62,194.33) and (432.94,192.66) .. (432.94,190.59) -- cycle ;

\draw (230.21,140.17) node [anchor=north west][inner sep=0.75pt]    {$\Omega _{0}$};
\draw (316.96,59.71) node [anchor=north west][inner sep=0.75pt]    {$\Omega _{1}$};
\draw (325.76,194.23) node [anchor=north west][inner sep=0.75pt]    {$\Omega _{2}$};
\draw (379.81,179.14) node [anchor=north west][inner sep=0.75pt]    {$\Omega _{3}$};
\draw (388.61,140.17) node [anchor=north west][inner sep=0.75pt]    {$\Omega _{4}$};
\draw (300.87,133.89) node [anchor=north west][inner sep=0.75pt]  [color={rgb, 255:red, 46; green, 77; blue, 152 }  ,opacity=1 ]  {$\textcolor[rgb]{0.18,0.3,0.6}{\Gamma _{12}}$};
\draw (307.15,165.26) node [anchor=north west][inner sep=0.75pt]    {$\textcolor[rgb]{0.25,0.46,0.02}{\mathbf{n}_{1}}$};
\draw (416.53,107.49) node [anchor=north west][inner sep=0.75pt]  [color={rgb, 255:red, 46; green, 77; blue, 152 }  ,opacity=1 ]  {$\textcolor[rgb]{0.18,0.3,0.6}{\Gamma }\textcolor[rgb]{0.18,0.3,0.6}{_{04}}$};

\end{tikzpicture}
\caption{Example of an admissible partition with $M=4$. Junction points are represented in red.}
\label{fig:Skeleton}
\end{figure}
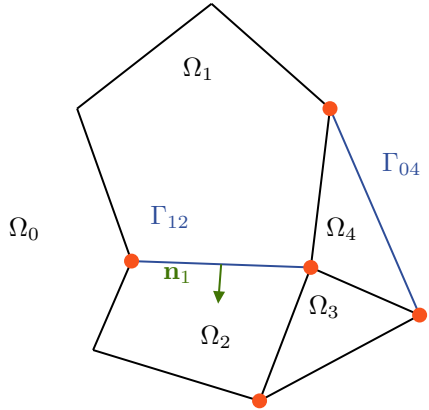

Introduced more than a decade ago, Multiple Trace Formulations (MTFs) are a class of BIEs that have proven particularly advantageous for addressing composite or heterogeneous wave scattering. They support both low-order and spectral discretization techniques \cite{hiptmairjerez2012multiplehelmholtz,claeys2014novel,JPT14}, facilitate the coupling of multiphysics problems \cite{HJA17,HJH18,MJP24}, and enable diagonal and operator preconditioning for complex geometries. These methods include the original local MTF \cite{hiptmairjerez2012multiplehelmholtz,HJL14, MR3069956}, the global version \cite{claeys2012multi}, and its variants \cite{coolsMTF,claeys2015quasi}. The term ``MTF" arises from treating each interior subdomain boundary trace as an independent unknown, thereby weakly enforcing the corresponding transmission conditions. In contrast, the PMCHWT can be regarded as a \emph{single trace formulation} (STF) in opposition to the MTF approach \cite{claeys2014novel,AYALA2022114356}. 
In acoustics and electromagnetics, all MTFs, except for the EM local MTF, have been proven to be well posed. However, while global MTF versions involve dense and computationally expensive cross-interaction terms, the local MTF leads to sparse interactions. Therefore, our focus will be exclusively on the local MTF.

The local MTF weakly enforces the Calder\'on identities within each subdomain and the transmission conditions across each interface. Let $\bu$ represent the vector of $M+1$ pairs of electric and magnetic boundary currents, and $\bff$ the corresponding source term. The local MTF results in a system of operator equations of the form:
\be \label{eq
} (2\bmA + \bmX) \bu = \bff, \ee where $\bmA:= \text{diag} (\hbmA_i)$ is a diagonal matrix operator containing the Calderón operators for each subdomain \eqref{eq:BOmultitraceI}, which are equivalent to Dirichlet-to-Neumann and Neumann-to-Dirichlet maps. The operator $\bm{\mathcal{X}}$ represents the sparse off-diagonal block transmission operator (see \Cref{rem:X}). Precise definitions will be provided later in \Cref{prob:MTF}. As mentioned, the transmission operator $\bmX$ enforces the transmission conditions in their strong form, which are expressed in weak form as duality products over interfaces $\Gamma_{ij}$:
$$
\int_{\Gamma_{ij}} \bu_i \cdot (\bn_{ij} \times \bv_j ) \d \Gamma_{ij},
$$
where $\bu_i$ is defined over $\Gamma_i$ and $\bv_j$ over $\Gamma_j$, respectively, requiring the suitable definition of boundary-to-interface restriction and extension-by-zero operators. Specifically, these operators map from $\Gamma_i$ to $\Gamma_{ij}$ and from $\Gamma_{ij}$ to $\Gamma_j$. One strength of the local MTF is that it can be accelerated through block preconditioning. For instance, one can leverage Calderón identities $(2\hbmA_i)^2 = \bmI_i$ \cite{betcke2020product}, or directly invert each submatrix $\mathbf{A}_i$ associated with the discretization of $\hbmA_i$ using LU factorization, as performed in the seminal paper \cite{hiptmairjerez2012multiplehelmholtz}. Yet, this can quickly become infeasible when dealing with three-dimensional wave scattering, as the submatrices are large, e.g., our examples in \Cref{subsec:complex} are of order $10^4$ degrees of freedom. Interestingly, a compromise can be achieved by applying a preconditioner $\bmP_i$ to each sub-block such that $\bmP_i \hbmA_i \approx \bmI_i $ remains a compact perturbation to the identity. 

As expected, the local MTF has a well-established connection with domain decomposition solvers, and in particular, with Schwarz methods \cite{CLAEYS201969, claeys2022nonlocal}. However, proving the stability of the local MTF in EM scattering remains a challenging task, despite promising results reported in \cite{AYALA2022114356}. Therein, the authors considered the case of a single subdomain ($M=1$), for which they demonstrated well-posedness and validated their findings through successful experiments on a non-smooth, non-convex domain---the Fichera cube. The challenges for the EM local MTF are twofold: (i) establishing theoretical well-posedness for $M>1$, and (ii) addressing practical issues in assembling the operator, including factors such as normal orientation, subdomain restrictions, and the handling of junction points. Based on the results in \cite{AYALA2022114356}, there is strong evidence that the local MTF is well-posed for multiple subdomains. However, a formal proof of this lies beyond the scope of the present manuscript.

In this work, we introduce, for the first time, a general local MTF for EM wave scattering and propose a novel framework to define spaces over subdomains and adjust the orientation of normal vectors based on the \emph{skeleton} (union of interfaces) mesh. This framework enables proper assembly of the local MTF operator and facilitates simulation of complex objects. Furthermore, we experimentally verify that the local MTF converges even in the presence of triple points. Instead of the standard multiplicative Calder\'on preconditioning, we present the first implementation of On-Surface-Radiation-Condition (OSRC) preconditioning  \cite{kriegsmann1987new} for MTF. Specifically, we employ the Padé-based OSRC operator preconditioning \cite{milinazzo1997rational,antoine2007generalized,el2014approximate,van2022frequency}, which achieves mesh-optimal and frequency-robust schemes \cite{van2022frequency}. This approach has been applied to the electric field integral equation \cite{fierro2023osrc}, acoustics \cite{van2022frequency} and elastodynamics \cite{chaillat2015approximate,darbas2015well}. While OSRC preconditioning may result in a comparable or slightly higher iteration count than standard Calder\'on-type counterparts, it is generally more computationally efficient as it does not require a barycentric refinement. Moreover, OSRC is especially effective for smooth domains and has demonstrated strong performance on closed domains with corners \cite{fierro2023osrc}. Additionally, block OSRC can be applied to the Calderón (or multitrace) operators $\hbmA_i$ defined in \eqref{eq:BOmultitraceI}, as described in \cite{van2022frequency} for acoustic scattering.

In addition to the above points, we investigate the behavior of the local MTF operator with an increasing number of artificial subdomains $M$, specifically focusing on how to preserve these desirable attributes through effective preconditioning strategies. The key practical findings are as follows:
\begin{enumerate}[label=(\roman*)]
    \item The local MTF for EM scattering performs well in conjunction with classical Rao-Wilton-Glisson (RWG) \cite{rao1982electromagnetic} or Raviart-Thomas \cite{raviart2006mixed} discrete function spaces;
\item The OSRC block preconditioning facilitates the development of convergent schemes for GMRES \cite{saad1986gmres};
\item As the number of domains $M$ increases, the execution times of the local MTF scales as $M^{2.20}$ (or $M^{1.72}$ by normalizing by total degrees of freedom), as shown in \Cref{tab:summaryFinal}. 
\end{enumerate}

The manuscript is structured as follows: In \Cref{sec:Problem setting}, we set notations and introduce the EM wave transmission problem for composite or heterogeneous materials. In \Cref{sec:BIEs}, we present the BIE framework and define the corresponding local MTF following \cite{hiptmairjerez2012multiplehelmholtz}. In \Cref{sec:Implementation}, we discuss the practical implementation for the MTF. Thus, in \Cref{sec:NumExp} we conduct extensive numerical experiments with increasingly complex objects, and future research directions are discussed in \Cref{sec:conclusions}.

\section{Problem setting}
\label{sec:Problem setting}
\subsection{Geometry and parameters}
\label{subsec:Notations}
For any $M \geq 1$, we consider the three-dimensional time-harmonic electromagnetic scattering by a bounded composite object $\Omega_S$, partitioning the free space into:
\be 
\IR^3   =   \bigcup_{i=0}^M \overline{\Omega}_i = \overline{\Omega}_0  \cup \overline{\Omega}_S,
\ee 
wherein, for each $i \in \{ 0,\ldots,M\}$, we assume that $\Omega_i$ is an open connected Lipschitz domain with boundary $\Gamma_i := \partial\Omega_i$, exterior unit normal $\bn_i$, and complement $\Omega_i^c:= \IR^3 \backslash \overline{\Omega}_i$. We define interfaces $\Gamma_{ij}:=\Gamma_i \cap \Gamma_j$, with unit normal $\bn_{ij}$ pointing from $\Gamma_i$ to $\Gamma_j$, for $i,j \in \{0,\ldots,M\}$, and introduce the index set \cite[Eq.~7]{hiptmairjerez2012multiplehelmholtz}:
$$
\Lambda_i := \left\{ j \in \{0,\ldots,M\} :  j \neq i \quad \text{and} \quad \Gamma_{ij} \neq \emptyset \right\}, \quad i \in \{0,\ldots,M\}.
$$
The skeleton is given by
\be
\label{eq:skeleton}
\Sigma :=  \bigcup_{i=0}^M \Gamma_i = \bigcup_{0 \leq i < j \leq M} \Gamma_{ij} =  \bigcup_{i=0}^M \left( \bigcup_{j \in \Lambda_i} \Gamma_{ij} \right).
\ee
We assume an angular frequency $\omega : = 2 \pi f$ for a time dependence $\exp(-\imath \omega t)$, $f > 0 $ and $\imath^2 = -1$. For $i \in \{ 0,\ldots,M\}$, we set each subdomain with constant permittivity and permeability $\epsilon_i ,\upmu_i$, along with wavenumbers $k_i:= \omega \sqrt{\upmu_i \epsilon_i}$ and an exterior wavelength $\lambda := \frac{2\pi}{k_0}$. Furthermore, let us set relative permittivities, permeabilities, impedances and admittances: 
\be \label{eq:relativeParameters}
\epsilon_{r,i} :=\frac{\epsilon_i}{\epsilon_0} , \quad \upmu_{r,i} :=\frac{\upmu_i}{\upmu_0} ,\quad  \eta_i := \sqrt{\frac{\upmu_{r,i}}{\epsilon_{r,i}}} \quad \text{and} \quad \varrho_i : = \sqrt{\frac{\epsilon_{r,i}}{\upmu_{r,i}}} = \frac{1}{\eta_i},
\ee
respectively.

\subsection{Functional spaces and traces}
\label{subsec:funcspaces}

For any $s>0$ and for $i \in \{0,\ldots,M\}$, we introduce the standard vector Sobolev spaces $\bsH^s(\Omega_i)$ (resp.~$\bsH^s(\Gamma_i)$) with $\bsL^2(\Omega_i):= \bsH^0(\Omega_i)$ (resp.~$\bsL^2(\Gamma_i) := \bsH^0(\Gamma_i)$) \cite{buffa2003galerkin}. For the unbounded exterior domain $\Omega_0$, we use the $\text{loc}$-subscript to denote spaces with bounded Sobolev norm over each compact subset $K \Subset \Omega_0$. For vectors, $\cdot$ and $|\cdot|$ will denote the Euclidean product and norm, respectively. We also make use of the subscript $\operatorname{pw}$ to signal piecewise spaces. For the $\bcurl$ being the curl operator, we recall the separable Hilbert space:
\be  \bsH(\bcurl ,\Omega_i) :=  \left\{\bU \in \bsL^2 (\Omega_i)  ~:~ \bcurl \bU \in \bsL^2(\Omega_i)\right\}
\ee
for $\Omega_i$ bounded and with its corresponding adaptation for $\Omega_0$. For any smooth vector field $\bU$ defined over $\Omega_i$, we introduce electric and magnetic traces \cite{buffa2003galerkin}:
\be
\bgamma_{D,i} \bU : = \bU|_{\Gamma_i} \times \bn_i \quad \text{and} \quad \bgamma_{N,i}\bU :=
 \frac{1}{\imath k_i} (\bcurl \bU |_{\Gamma_i} \times \bn_i)=\frac{1}{\imath k_i}\bgamma_{D,i} \bcurl \bU.
\ee
Furthermore, let the $c$-superscript refer to traces taken with respect to $\Omega_i^c$ for $i \in \{0,\ldots,M\}$. Jumps and averages of traces for $\cdot  \in \{D,N\}$ are set as:
\be 
\left[\bgamma_{\cdot,i}\right]_{\Gamma_i} := \bgamma_{\cdot,i}^c - \bgamma_{\cdot,i} \quad \text{and} \quad \left\{ \bgamma_{\cdot,i}  \right\}_{\Gamma_i} := \frac{1}{2}\left(\bgamma^c_{\cdot,i} + \bgamma_{\cdot,i}\right) .
\ee
For $i \in \{0,\ldots,M\}$, and  $\divg_{\Gamma_i}$ being the tangential divergence operator, we define the spaces \cite{buffa2003galerkin}
\be\label{eq:XGammaspace}
 \boldsymbol{X}(\Gamma_i) : = \bsH_\times^{-1/2}(\divg_{\Gamma_i}, \Gamma_i)
\ee 
of tangential traces on the boundary. The mapping $\bgamma_{D,i} : \bsH(\bcurl, \Omega_i) \to \boldsymbol{X}(\Gamma_i) $ is continuous and surjective. Moreover, for $\bu, \bv \in \boldsymbol{X}(\Gamma_i)$ the self-duality for tangential traces reads as \cite[Eq.~10]{buffa2003galerkin}:
\be \label{eq:self_duality_identity}
\langle\bu, \bv  \rangle_{\times,i} \equiv  \langle \mI^i\bu, \bv  \rangle_{\times,i}:= \int_{\Gamma_i} \bu \cdot (\bn_i \times \bv )\d \Gamma_i.
\ee
with the identity operator $\mI^i: \boldsymbol{X}(\Gamma_i) \to \boldsymbol{X}(\Gamma_i)$.

For  $i \in \{0 , \ldots,M\}$, we define the \emph{Cauchy trace space} of electric and magnetic surface currents 
$$
\IH (\Gamma_i ) :=  \boldsymbol{X}(\Gamma_i) \times \boldsymbol{X}(\Gamma_i),
$$
with $\boldsymbol{X}(\Gamma_i)$ as in \eqref{eq:XGammaspace}. Finally, let us set the \emph{multiple traces space}:
\be 
\quad \IH(\Sigma) : = \IH(\Gamma_0) \times  \ldots \times \IH(\Gamma_M),
\ee 
which is self-dual with respect to the pairing:
\be 
 \left\langle    \bigoplus_{i=0}^M  \begin{pmatrix} \bu_i \\ \bp_i \end{pmatrix} ,   \bigoplus_{i=0}^M \begin{pmatrix}\bv_i \\ \bq_i \end{pmatrix}  \right\rangle_{\times} : =  \sum_{i=0}^M \Big(\langle \bu_i,\bq_i\rangle_{\times,i} + \langle \bp_i , \bv_i\rangle_{\times,i} 
 \Big).
\ee
\subsection{Heterogenous EM wave transmission problem}
\label{subsec:TransmissionProblem}
We now introduce the Maxwell wave transmission problem over composite scatterers. To guarantee well-posedness in the case of bounded obstacles, we require the Silver-M\"uller radiation condition \cite[Eq.~3]{buffa2003boundary} denoted by SM$(\cdot)$. 
\begin{problem}
\label{prob:transmission}
Consider a penetrable scatterer $\Omega_S\subset\mathbb{R}^3$ with boundary $\Gamma$ composed of $M\geq1$ homogeneous bounded subdomains $\Omega_i$, $i=1,\ldots,M$, and exterior $\Omega_0$. Assume an incident field $\bE^\textup{inc}$ such that $\bcurl \bcurl \bE^\textup{inc} - k_0^2 \bE^\textup{inc} =  0$ in $\Omega_0$. We seek the total field $\bE
\in\bsH_\textup{loc}(\bcurl, \mathbb{R}^3\setminus\Gamma)$, with scattered field $\bE^\textup{sc}:=\bE-\bE^\textup{inc} \in \bsH_\textup{loc}(\Omega_0)$
such that 
\begin{subequations}
\begin{align}
\label{eq:curlcurl}\bcurl \bcurl \bE - k_i^2 \bE &=  0  &\textup{in} \quad  \Omega_i,~i =0,\ldots,M,\\
  \bgamma_{D,i}  \bE &= -  \bgamma_{D,j} \bE   & \textup{on} \quad \Gamma_{ij},~i\in \{0,\ldots,M\} , ~j \in \Lambda_i ,\\
 \varrho_i \bgamma_{N,i} \bE   &= -\varrho_j \bgamma_{N,j} \bE  &  \textup{on} \quad \Gamma_{ij},~i \in \{0,\ldots,M\} ,~j \in \Lambda_i,\\[4pt]
 \textup{SM}(\bE^\textup{sc})& & \textup{for}\quad |\bx| \to \infty.
 \end{align} 
 \end{subequations}
\end{problem}
In what follows, we will define $\bE^i:=\bE|_{\Omega_i}$ for $i \in \{0, \ldots, M\}$.
\begin{remark}
Neumann (magnetic) transmission conditions on $\Gamma_{ij}$ are:
\begin{align*} 
\frac{1}{\upmu_i} (\bcurl \bE|_{\Gamma_i} \times \bn_i)   &= -\frac{1}{\upmu_j
} (\bcurl \bE|_{\Gamma_j}\times \bn_j)  \\
\frac{k_i}{\upmu_i} \bgamma_{N,i}\bE   &= -\frac{k_j}{\upmu_j
} \bgamma_{N,j} \bE.
\end{align*}
However, 
$$
\frac{k_i}{\upmu_i} = \omega \frac{\sqrt{\epsilon_i \upmu_i}}{\upmu_i} = \omega \sqrt{\frac{\epsilon_i}{\upmu_i}} = \omega \varrho_i,
$$
yielding
\be 
\varrho_i\gamma_{N,i} \bE= - \varrho_j\gamma_{N,j} \bE.
\ee
\end{remark}

\subsection{Boundary integral potentials and operators}
\label{subsec:IntegralOperators}
Since we are seeking the EM fields in an unbounded domain, we will reduce the dimensionality of the problem by recasting the original volume boundary value Problem \ref{prob:transmission} into a system of BIEs. To that end, we introduce the corresponding EM boundary potentials and integral operators (BIOs). 

For $i\in \{0,\ldots,M\}$ and any $k>0$ the electric and magnetic potentials are defined as:
\be \label{eq:potential_operators}
\begin{split}
\mE_{k}^i \bv (\bx )& : = \imath k \int_{\Gamma_i} G_{k} (\bx, \by)\bv ( \by)  \d \Gamma_i (\by)  - \frac{1}{\imath k} \nabla_\bx \int_{\Gamma_i}   G_{k}(\bx,\by)\divg_\by \cdot \bv(\by) \d \Gamma_i (\by), \\
\mH_{k}^i \bv (\bx ) & : = \bcurl \int_{\Gamma_i} G_{k} (\bx, \by)\bv ( \by)  \d \Gamma_i (\by) ,
\end{split}
\ee 
wherein $G_k(\bx, \by) = \dfrac{\exp ( \imath k |\bx - \by| )}{4 \pi |\bx - \by| } $ is the standard scalar Helmholtz kernel and $\divg$ is the divergence operator. By the Stratton-Chu formulae \cite{kirsch2015mathematical}, any interior electric field $\bU^{i}\in\bsH_\textup{loc}(\bcurl,\Omega_i)$ satisfying the homogeneous Maxwell equation \eqref{eq:curlcurl} can be represented as \cite[Eq.~14]{KLEANTHOUS2022111099}
\begin{alignat}{3}
\label{eqn:StrattonChu_int}
\mathcal{H}^i_{k_i} (\gamma_{D,i} \bU^{i}) &+ \mathcal{E}^i_{k_i} (\gamma_{N,i} \bU^{i}) = 
\begin{cases}
\bU^{i}(\mathbf{x})  , & \mathbf{x} \in \Omega_i, \\
\mathbf{0}, & \mathbf{x} \in \Omega^c_i
\end{cases}.
\end{alignat}
Next, we introduce the electric and magnetic field boundary integral operators $\boldsymbol{X}(\Gamma_i) \to \boldsymbol{X}(\Gamma_i)$:
\be \label{eq:boundary_operators}
\begin{cases}
\text{EFIO:} & \mT_{k}^i  : = \left\{ \bgamma_{D,i} \mE^i_{k} \right\}_{\Gamma_i}  = - \left\{ \bgamma_{N,i} \mH^i_{k}\right\}_{\Gamma_i} , \\
\text{MFIO:} & \mK_{k}^i : = \left\{ \bgamma_{D,i} \mH^i_{k} \right\}_{\Gamma_i} = \left\{ \bgamma_{N,i} \mE_{k}^i\right\}_{\Gamma_i}. \\
\end{cases}
\ee 

For $i \in \{0,\ldots,M\}$, we also define the multiple traces identity operator along with the scaled multiple traces one
\be \label{eq:BOmultitraceI}
\bmI_i : = \begin{bmatrix} \mI^i &0 \\ 0 & \mI^i  \end{bmatrix}\quad \text{and} \quad  \hbmA^i_k \equiv \hbmA_i: = \begin{bmatrix}  \mK^i_k  & \varrho_i^{-1} \mT^i_k   \\ - \varrho_i \mT^i_k  & \mK^i_k  \end{bmatrix},
\ee 
mapping from $\IH (\Gamma_i )$ into itself and with $\mI^i$ as in \eqref{eq:self_duality_identity}. 

For $i \in \{0,\ldots,M\}$, we introduce the \emph{Cauchy data} for the solution and incident field in \Cref{prob:transmission} as:
\begin{align}
    \label{eqn:boundary_conditions_traces} \mathbf{u}_i :=  \begin{bmatrix}
\gamma_{D,i} \bE\\[6pt]
\varrho_i \gamma_{N,i} \bE
\end{bmatrix} \quad \text{and} \quad  \mathbf{u}^\text{inc}_0 :=  \begin{bmatrix}
\gamma_{D,0} \bE^{\text{inc}}\\[6pt]
\gamma_{N,0} \bE^{\text{inc}}
\end{bmatrix},
\end{align}
yielding $\bu , \bu^\text{inc} \in \IH(\Sigma)$ with
\be \label{eq:bu_buinc}
\bu = \begin{bmatrix} \bu_0 \\  \bu_1 \\ \vdots\\ \bu_M \end{bmatrix} \quad \text{and} \quad \bu^\text{inc} = \begin{bmatrix}   \bu^\text{inc}_0 \\ 0 \\ \vdots \\ 0 \end{bmatrix}.
\ee 
Setting $\bu^\text{sc} := \bu - \bu^\text{inc}$, the (interior) Calderón projector property \cite[Eq.~50]{buffa2003boundary} for each domain $i \in \{ 0,\ldots,M\}$ is expressed as:
\be \label{eq:interiorCalderon}
   \left(\frac{1}{2}\bmI_i  - \hbmA_i  \right) \bu^\text{sc} = 0.
\ee 
Notice that the formula remains valid for $i = 0$, as $\Omega_0$ is defined with the normal vector $\bn_0$ pointing towards the exterior of the domain. Next, \eqref{eq:interiorCalderon} can be rewritten as:
\be \label{eq:cald_temp}
 2 \hbmA_i   \bu^\text{sc}  - \bu^\text{sc} = 0.
\ee
In $\Omega_0$, the interior Calderón identity for the scattered field is:
\be \label{eq:cald1}
\left(\frac{\bmI_0}{2} -  \hbmA_0 \right) \bu_0^\text{sc} = \bu_0^\text{sc}.
\ee
Next, we succinctly present the Calderón identities  in \Cref{prop:calderon_identities}.
\begin{proposition}[Calderón identities]\label{prop:calderon_identities}
Let $\bu$ and $\bu^\textup{inc}$ be defined as in \eqref{eqn:boundary_conditions_traces}. For $i \in \{ 0,\ldots,M\}$, there holds that
\be\label{eq:relCalderon} 
2  \hbmA_i  \bu_i - \bu_i = 2 \bu^\textup{inc}_i. 
\ee 
\end{proposition}

\begin{proof}Refer to \ref{app:appendixA}.
\end{proof}

Finally, for $i \in \{0,\ldots,M\}$, we introduce the transmission operators across interfaces: 
\be \label{eq:BOTransmission}
\renewcommand\arraystretch{1.3}  
 \bmI_{ij} : = \begin{bmatrix} \tmI^{ij} &0 \\ 0 & \tmI^{ij}  \end{bmatrix}\quad  \text{for} \quad  j\in \Lambda_i  \quad\text{and}\quad \bmI_{ij} : = \begin{bmatrix} 0 &0 \\ 0 & 0 \end{bmatrix} \quad \text{else,}
\ee 
where $\tmI^{ij}: \boldsymbol{X}_\text{pw}(\Gamma_i) \to \boldsymbol{X}_\text{pw}(\Gamma_j)$ is the restriction-and-extension-by-zero operator defined as:
\be \label{eq:tildeX}
\tmI^{ij}:=
\begin{cases}
  1  \quad \text{for} \quad \bx  \in \Gamma_{ij},\\
0 \quad \text{elsewhere on } \Gamma_{ij}.
\end{cases}
\ee 
We use the tilde superscript to distinguish the transmission operator $\tmI^{ij}: \boldsymbol{X}_\text{pw}(\Gamma_i) \to \boldsymbol{X}_\text{pw}(\Gamma_j)$ from the identity operator over $\Gamma_{ij}$ being  $\mI^{ij}: \boldsymbol{X}(\Gamma_{ij}) \to \boldsymbol{X}(\Gamma_{ij})$.

\section{Local multiple traces formulation}
\label{sec:BIEs}
In what follows, we use the previous tools to state the continuous and discrete versions of the Maxwell local MTF. As mentioned, the idea is to set electric and magnetic boundary current pairs (Cauchy data) as independent unknowns per subdomain, each current pair satisfying its corresponding Calder\'on identity and coupled to others via weak transmission conditions.
\subsection{Continuous setting}
The local MTF is obtained by appropriately combining the Calderón identities in \Cref{prop:calderon_identities} and the transmission conditions. 
Notice that the transmission conditions for  $i \in \{0,\ldots,M\}$ can be rewritten as
\be\label{eq:transmission_for_mtf}
\bu_i = \sum_{j \in \Lambda_i} \bu_i|_{\Gamma_{ij}} =  -\sum_{j \in \Lambda_i}  \bmI_{ji} \bu_j.
\ee 
yielding the local MTF by combining \eqref{eq:transmission_for_mtf} and \eqref{eq:relCalderon} 
\be 
2  \hbmA_i  \bu_i + \sum_{j \in \Lambda_i} \bmI_{ji} \bu_j = 2 \bu^\textup{inc}_i.
\ee
Observe that now the factors $\rho_i$ show up in the operators $\hbmA_i$. This differs from the original Helmholtz formulation provided \cite{hiptmairjerez2012multiplehelmholtz}.

\begin{problem}[Maxwell MTF problem]\label{prob:MTF}
For $\bu^\textup{inc} \in \IH(\Sigma)$, seek $\bu \in \IH(\Sigma)$ such that the variational form:
\be 
\langle \bfcal{M} \bu, \bv \rangle_\times = \langle \bu^\textup{inc} , \bv \rangle_\times , \quad \forall \bv\ \in \widetilde{\IH}(\Sigma),
\ee 
is satisfied with
\be 
\bfcal{M} := 2 \bmA + \bmX =  \begin{bmatrix}
    2\hbmA_0 & \bmI_{10} & \bmI_{20} & \ldots & \bmI_{M0} \\[3pt]
    \bmI_{01} & 2\hbmA_1 & \bmI_{21}  &  \ldots & \bmI_{M1}\\[3pt]
    \vdots & \vdots & \vdots & \ddots & \vdots \\[3pt]
    \bmI_{0M} & \bmI_{1M} & \bmI_{2M} & \ldots & 2 \hbmA_{M} 
    \end{bmatrix}.
\ee 
wherein $\bmA : = \textup{diag}(\hbmA_i)$ and $\bmX : = (\bmI_{ji})$ for $i,j \in \{0,\ldots,M\}$.
\end{problem}

\begin{remark}
    Problem \ref{prob:MTF} can be shown to be well posed and equivalent to \Cref{prob:transmission} in the continuous case.
\end{remark}

\begin{remark}\label{rem:X}
    For clarity, note that Helmholtz and Laplace MTF problems (refer for example to \cite[Problem 6]{hiptmairjerez2012multiplehelmholtz}) are typically presented in the form $2 \bmA -\bmX$, where $\bmX:= (\bmX_{ji})$ for $i,j\in \{0, \ldots,M\}$, with 
    \be 
    \bmX_{ij} : = \begin{bmatrix} \tmI^{ij} &0 \\ 0 & -\tmI^{ij}  \end{bmatrix}.
    \ee
    In the Maxwell case, both the electric and magnetic traces involve the normal vector, resulting in a minus sign for both traces. In contrast, for Laplace and Helmholtz problems, the minus sign appears only for the Neumann trace.
\end{remark}
\begin{remark}[PMCHWT]
The PMCHWT (or STF) for $M=1$ can be derived from the local MTF formulation by summing the first and second rows in \Cref{prob:MTF}:
\be
2\hbmA_0 \bu_0  +  \bu_1 + \bu_0  + 2 \hbmA_1 \bu_1 =2 \bu_0^\textup{inc}
\ee 
yielding for the single trace unknown $\bu$, and using relation $\bu_0=-\bu_1$:
\be 
 (\hbmA_0 -  \hbmA_1)\bu_0  = \bu^\textup{inc}.
\ee
\end{remark}
\subsection{Discretization setting}
Recall the skeleton definition \eqref{eq:skeleton} in the sense of a union of disjoint oriented surfaces
$$
 \Sigma = \bigcup_{i \in \{0,\ldots,M\} ,j\in \Lambda_{i}} \Gamma_{ij} .
$$
with each $\Gamma_{ij}$ oriented with $\bn_{ij}$. Under this setting, each $\Gamma_{ij}$ is discretized using triangular (possibly curved) elements yielding a conforming mesh $\Gamma_{ij}^h$, with meshsize $h>0$. We assume that the normal vector for each interface is from $\min(i,j)$ to $\max(i,j)$. The union of these meshes issues the mesh for $\Sigma^h$ of the skeleton. 
\begin{remark}[Geometries and subdomains]\label{rmk:gmsh}
Meshes are typically generated from a specified geometry parameterization (for example, a $\texttt{.geo}$ in GMSH \cite{geuzaine2009gmsh}). Specifically focusing on GMSH, which is praised for its simplicity and versatility, one can utilize the $\texttt{Physical Surface}$ feature to correctly identify each $\Gamma_{ij}$ and subsequently generate a conforming mesh for the skeleton. Note that the use of a minus sign enables the selection of the normal vector orientation definition when meshing $\texttt{Physical Surface}$.
\end{remark}
\begin{remark}[Conventions for meshing $\Gamma_{ij}$]
The procedure assumes that each $\Gamma_{ij}^h$ is meshed with rule $i<j$ for simplicity and to establish consistent a convention. However, the approach is flexible with respect to interface indexing. For example, in the case of triple a point (middle column in \Cref{tab:DomainConfiguration}), we may mesh $\Gamma_{10}^h$, $\Gamma_{12}^h$ and $\Gamma_{20}^h$. Therefore, we have $\Gamma_0 = \texttt{[-10,-20]}$, $\Gamma_1 = \texttt{[10,12]}$ and $\Gamma_2 = \texttt{[20,-12]}$. This configuration aligns more naturally with classic formulations, as it is resembles the meshing approach for scatterers $\Omega_1$ and $\Omega_2$, where boundaries involving $\Omega_0$ are taken from the scatterer towards the exterior domain.
\end{remark}

We introduce the space RWG$_i$ of cardinality $N_i$ to be the span of the Rao-Wilton-Glisson (RWG) basis functions \cite{rao1982electromagnetic} given by the direct sum:
\be \label{eq:rwg_multitrace}
\IH_h: =\bigoplus_{i=0}^M \left(\begin{array}{c}{\rm RWG}_i \\ {\rm RWG}_i \end{array}\right) \subset \IH(\Sigma), \quad \text{with}\quad N:=\dim(\IH_h)= 2\sum_{i=0}^M N_i.
\ee 
Similarly, we can define the space RWG$_{ij}$ over interfaces as the restriction of RWG$_i$ to the interface $\Gamma_{ij}$, where the normal vector $\bn_{ij}$ points from $\Omega_i$ to $\Omega_j$. Notably, multiple-trace spaces involve a doubling of unknowns on each interface, as they include a pair of Cauchy traces taken from both $\Omega_i$ and $\Omega_j$. In contrast, single-trace spaces typically contain only one pair of Cauchy traces per interface. However, this doubling of unknowns makes the local MTF operator well-suited for block preconditioning within each subdomain. 

Let us call the basis elements of $\IH_h$ $\{\bpsi_n\}_{n=1}^N$, then the Galerkin solution $\bu_h = \sum_{n=1}^N u_n \bpsi_n \in \IH_h$ satisfies the following variational problem:
\be 
\langle \bmM \bu_h, \bv_h \rangle_\times = \Big{\langle} \bu_0^\text{inc} , \bv_h \Big{\rangle}_\times , ~ \forall \ \bv_h \in \IH_h.
\ee 
One can construct the Galerkin MTF matrix $\bM:=( \bM_{mn})$ for $m,n\in \{ 1, 
\ldots ,N\}$ via the following duality products
\be
\bM_{mn} = \langle \bmM \bpsi_n, \bpsi_m \rangle_\times, \quad m,n= 1, \ldots, N.
\ee

\section{Practical implementation of the MTF}\label{sec:Implementation}
We now discuss the practical aspects that facilitate the computation of the local MTF operator. It is worth mentioning that these apply equally to both EM and acoustic (including Laplace) scattering problems.
\subsection{Motivation}
As mentioned in \Cref{sec:intro}, the local MTF involves:
\begin{enumerate}
    \item[(i)] Calderón operators $\hbmA_i$ \eqref{eq:BOmultitraceI};
    \item[(ii)] Transmission operators $\bmI_{ij}$ \eqref{eq:BOTransmission}.
\end{enumerate}
Interestingly, the matrix coefficients for these operators are generally easy to access through BEM solvers. In fact, for $i\in \{0,\ldots,M\}$, one can mesh $\Gamma_i$ and assembly the matrix for $\hbmA_i$. In the same fashion, the coefficients for the discretization of $\bmI_{ij}$ can be derived using suitable dof-mappings from the data in $\bmI_{i}$. However, this approach is insufficient for handling complex settings. Transferring the degrees of freedom (dofs) for all these matrices to the multiple trace discrete spaces and mapping the extension-to-restriction operator is far from straightforward. To the authors' knowledge, most local MTF BEM routines manually perform these dof mappings, leading to non-automated schemes, typically limited to a few objects. For scalable and automatic routines, the information between subdomains must be transferred to a common data receptor. This is achieved using a top-down approach centered on the skeleton mesh, which encompasses all the subdomains and acts as the \emph{core information dictionary} for the space.

\subsection{Top-down approach for space definitions}
Consider the skeleton mesh $\Sigma^h$ in \Cref{sec:BIEs}. To properly define the discretization spaces in \eqref{eq:rwg_multitrace}, it is necessary to assemble operators on the triangles associated with the given subdomains and to set the corresponding the normal vector. For each element $\tau \in \Sigma_h$, we attach a normal orientation multiplier $\texttt{swap} \in \{-1,1\}$ and its degrees of freedom $\texttt{dofs}$ (e.g.~the dofs associated to each edge for RWG function space). The normal multipliers allows to swap the normal vector orientation in assembly routines.

Note that with the current setting,
$$
\Gamma_i = \bigcup_{j \in \Lambda_i} \Gamma_{ij}.
$$
\begin{remark}
    Interestingly, the $\Lambda_i$ set can be represented as a connectivity graph. Each subdomain $\Gamma_i$ for $i\in \{0,\ldots,M\}$ is defined as a node, and an edge is added between nodes $i$ and $j$ if $j \in \Lambda_i$.
\end{remark}
Since the the normal vector is defined from $\min(i,j)$ to $\max(i,j)$, no action is needed when $\min(i,j) = i$ (i.e., $\texttt{swap}=1$), while the normal is swapped when $\min(i,j)=j$ (i.e., $\texttt{swap} = -1$). This is detailed in \Cref{alg:process_segment} and ensures the correct definition of RWG$_i$ spaces. Additionally, the normal swapping vector NM$_i$ can be easily integrated into the core assembly routines. 

\begin{algorithm} \caption{Process subdomains and swap normals}
\begin{algorithmic}[1]
\label{alg:process_segment}
\Require Skeleton mesh: $\Sigma_h$; index: $i \in 0 , \ldots , M$. 
\Repeat
\ForEach {$ j \in \Lambda_i $}
\State $\texttt{swap} \gets    \sign ( j - i) $
\State $\dofs \gets$   \{dofs associated to triangles $\tau \in \Gamma^h_{ij}$ \} 
\State RWG$_i \gets$  RWG$_i \cup  \dofs $ 
\State NM$_i \gets$  NM$_i~\cup $ \{$\texttt{swap}$\} 
\Until {}
\Ensure RWG$_i$, NM$_i$
\end{algorithmic}
\end{algorithm}
With \Cref{alg:process_segment} at hand, the Calderón operator can be properly defined. With regards to transmission operators, the last step is to handle the restriction-to-extension process \eqref{eq:tildeX}. Set $i\in\{0,\ldots,M\}$ and $j\in\Lambda_i$. For any $\bphi_n \in \text{RWG}_i$, $n \in \{1 , \ldots, N_i\}$ and $\bpsi_m \in \text{RWG}_j$, $m \in \{1, \ldots, N_j\}$,  the matrix for the transmission operator in \eqref{eq:BOTransmission} reads as
\be\label{eq:trick_transmission}
\begin{split}
\widetilde{\bI}^{ij}_{mn}& =   \langle  \tmI^{ij} \bphi_n,  \bpsi_m  \rangle_\times =  \sum_{\tau \in \Gamma^h_{ij}}  \bphi_n  \cdot (\bn_{ij} \times \bpsi_m ) =  \langle  \mI^{i} \bphi_n, \bpsi_m\rangle_\times, 
\end{split}
\ee
for all $n=1, \ldots, N_i$, $m=1, \ldots , N_j$.
This formula shows that the transmission operator can be obtained without modifying traditional BEM solvers, by assembling the classical $\texttt{identity}$ operator over RWG$_i$ as the domain and RWG$_j$ as the dual-to-range. For clarity, we examine several relevant cases in detail in \Cref{tab:DomainConfiguration}. 

\begin{table}[htb!]
\renewcommand\arraystretch{1.4}
\begin{center}
\footnotesize
\begin{tabular}{
    |>{\centering\arraybackslash}m{2cm}
    |>{\centering\arraybackslash}m{4cm}
    |>{\centering\arraybackslash}m{4cm}
    |>{\centering\arraybackslash}m{4cm}|}
\hline
\bf Case &  
\bf Dielectric problem & \bf Triple point  & \bf Grid problem \\ \hhline{|=|=|=|=|}
&&&\\
Simplified representation &  \tikzset{every picture/.style={line width=0.75pt}} 

\begin{tikzpicture}[x=0.75pt,y=0.75pt,yscale=-1,xscale=1]

\draw   (379,67.5) -- (430,67.5) -- (430,116.5) -- (379,116.5) -- cycle ;
\draw   (441,67.5) -- (492,67.5) -- (492,116.5) -- (441,116.5) -- cycle ;

\draw (384,37.17) node [anchor=north west][inner sep=0.75pt]    {$\Omega _{0}$};
\draw (393.54,80.44) node [anchor=north west][inner sep=0.75pt]    {$\Omega _{1}$};
\draw (455.6,80.44) node [anchor=north west][inner sep=0.75pt]    {$\Omega _{2}$};

\end{tikzpicture} &  \tikzset{every picture/.style={line width=0.75pt}} 

\begin{tikzpicture}[x=0.75pt,y=0.75pt,yscale=-1,xscale=1]

\draw   (564,64.5) -- (615,64.5) -- (615,113.5) -- (564,113.5) -- cycle ;
\draw   (615,64.5) -- (666,64.5) -- (666,113.5) -- (615,113.5) -- cycle ;
\draw  [draw opacity=0][fill={rgb, 255:red, 248; green, 82; blue, 28 }  ,fill opacity=1 ] (611.26,113.5) .. controls (611.26,111.43) and (612.93,109.76) .. (615,109.76) .. controls (617.07,109.76) and (618.74,111.43) .. (618.74,113.5) .. controls (618.74,115.57) and (617.07,117.24) .. (615,117.24) .. controls (612.93,117.24) and (611.26,115.57) .. (611.26,113.5) -- cycle ;
\draw  [draw opacity=0][fill={rgb, 255:red, 248; green, 82; blue, 28 }  ,fill opacity=1 ] (611.26,64.5) .. controls (611.26,62.43) and (612.93,60.76) .. (615,60.76) .. controls (617.07,60.76) and (618.74,62.43) .. (618.74,64.5) .. controls (618.74,66.57) and (617.07,68.24) .. (615,68.24) .. controls (612.93,68.24) and (611.26,66.57) .. (611.26,64.5) -- cycle ;

\draw (569,34.17) node [anchor=north west][inner sep=0.75pt]    {$\Omega _{0}$};
\draw (578.54,77.44) node [anchor=north west][inner sep=0.75pt]    {$\Omega _{1}$};
\draw (629.6,77.44) node [anchor=north west][inner sep=0.75pt]    {$\Omega _{2}$};

\end{tikzpicture} & \tikzset{every picture/.style={line width=0.75pt}} 

\begin{tikzpicture}[x=0.75pt,y=0.75pt,yscale=-1,xscale=1]

\draw   (750,63.5) -- (801,63.5) -- (801,112.5) -- (750,112.5) -- cycle ;
\draw   (801,63.5) -- (852,63.5) -- (852,112.5) -- (801,112.5) -- cycle ;
\draw   (801,112.5) -- (852,112.5) -- (852,161.5) -- (801,161.5) -- cycle ;
\draw   (750,112.5) -- (801,112.5) -- (801,161.5) -- (750,161.5) -- cycle ;
\draw  [draw opacity=0][fill={rgb, 255:red, 248; green, 82; blue, 28 }  ,fill opacity=1 ] (746.26,112.5) .. controls (746.26,110.43) and (747.93,108.76) .. (750,108.76) .. controls (752.07,108.76) and (753.74,110.43) .. (753.74,112.5) .. controls (753.74,114.57) and (752.07,116.24) .. (750,116.24) .. controls (747.93,116.24) and (746.26,114.57) .. (746.26,112.5) -- cycle ;
\draw  [draw opacity=0][fill={rgb, 255:red, 248; green, 82; blue, 28 }  ,fill opacity=1 ] (797.26,112.5) .. controls (797.26,110.43) and (798.93,108.76) .. (801,108.76) .. controls (803.07,108.76) and (804.74,110.43) .. (804.74,112.5) .. controls (804.74,114.57) and (803.07,116.24) .. (801,116.24) .. controls (798.93,116.24) and (797.26,114.57) .. (797.26,112.5) -- cycle ;
\draw  [draw opacity=0][fill={rgb, 255:red, 248; green, 82; blue, 28 }  ,fill opacity=1 ] (797.26,63.5) .. controls (797.26,61.43) and (798.93,59.76) .. (801,59.76) .. controls (803.07,59.76) and (804.74,61.43) .. (804.74,63.5) .. controls (804.74,65.57) and (803.07,67.24) .. (801,67.24) .. controls (798.93,67.24) and (797.26,65.57) .. (797.26,63.5) -- cycle ;
\draw  [draw opacity=0][fill={rgb, 255:red, 248; green, 82; blue, 28 }  ,fill opacity=1 ] (746.26,63.5) .. controls (746.26,61.43) and (747.93,59.76) .. (750,59.76) .. controls (752.07,59.76) and (753.74,61.43) .. (753.74,63.5) .. controls (753.74,65.57) and (752.07,67.24) .. (750,67.24) .. controls (747.93,67.24) and (746.26,65.57) .. (746.26,63.5) -- cycle ;
\draw  [draw opacity=0][fill={rgb, 255:red, 248; green, 82; blue, 28 }  ,fill opacity=1 ] (848.26,63.5) .. controls (848.26,61.43) and (849.93,59.76) .. (852,59.76) .. controls (854.07,59.76) and (855.74,61.43) .. (855.74,63.5) .. controls (855.74,65.57) and (854.07,67.24) .. (852,67.24) .. controls (849.93,67.24) and (848.26,65.57) .. (848.26,63.5) -- cycle ;
\draw  [draw opacity=0][fill={rgb, 255:red, 248; green, 82; blue, 28 }  ,fill opacity=1 ] (848.26,112.5) .. controls (848.26,110.43) and (849.93,108.76) .. (852,108.76) .. controls (854.07,108.76) and (855.74,110.43) .. (855.74,112.5) .. controls (855.74,114.57) and (854.07,116.24) .. (852,116.24) .. controls (849.93,116.24) and (848.26,114.57) .. (848.26,112.5) -- cycle ;
\draw  [draw opacity=0][fill={rgb, 255:red, 248; green, 82; blue, 28 }  ,fill opacity=1 ] (746.26,161.5) .. controls (746.26,159.43) and (747.93,157.76) .. (750,157.76) .. controls (752.07,157.76) and (753.74,159.43) .. (753.74,161.5) .. controls (753.74,163.57) and (752.07,165.24) .. (750,165.24) .. controls (747.93,165.24) and (746.26,163.57) .. (746.26,161.5) -- cycle ;
\draw  [draw opacity=0][fill={rgb, 255:red, 248; green, 82; blue, 28 }  ,fill opacity=1 ] (797.26,161.5) .. controls (797.26,159.43) and (798.93,157.76) .. (801,157.76) .. controls (803.07,157.76) and (804.74,159.43) .. (804.74,161.5) .. controls (804.74,163.57) and (803.07,165.24) .. (801,165.24) .. controls (798.93,165.24) and (797.26,163.57) .. (797.26,161.5) -- cycle ;
\draw  [draw opacity=0][fill={rgb, 255:red, 248; green, 82; blue, 28 }  ,fill opacity=1 ] (848.26,161.5) .. controls (848.26,159.43) and (849.93,157.76) .. (852,157.76) .. controls (854.07,157.76) and (855.74,159.43) .. (855.74,161.5) .. controls (855.74,163.57) and (854.07,165.24) .. (852,165.24) .. controls (849.93,165.24) and (848.26,163.57) .. (848.26,161.5) -- cycle ;

\draw (755,33.17) node [anchor=north west][inner sep=0.75pt]    {$\Omega _{0}$};
\draw (764.54,76.44) node [anchor=north west][inner sep=0.75pt]    {$\Omega _{1}$};
\draw (764.54,124.19) node [anchor=north west][inner sep=0.75pt]    {$\Omega _{2}$};
\draw (815.6,124.19) node [anchor=north west][inner sep=0.75pt]    {$\Omega _{3}$};
\draw (815.6,76.44) node [anchor=north west][inner sep=0.75pt]    {$\Omega _{4}$};

\end{tikzpicture}\\ \hline
&&&     \\
Connectivity graph &  \tikzset{every picture/.style={line width=0.75pt}} 

\begin{tikzpicture}[x=0.75pt,y=0.75pt,yscale=-1,xscale=1]

\draw  [color={rgb, 255:red, 0; green, 0; blue, 0 }  ,draw opacity=1 ] (30.5,76.25) .. controls (30.5,63.55) and (40.8,53.25) .. (53.5,53.25) .. controls (66.2,53.25) and (76.5,63.55) .. (76.5,76.25) .. controls (76.5,88.95) and (66.2,99.25) .. (53.5,99.25) .. controls (40.8,99.25) and (30.5,88.95) .. (30.5,76.25) -- cycle ;
\draw  [color={rgb, 255:red, 0; green, 0; blue, 0 }  ,draw opacity=1 ] (112.36,48.5) .. controls (112.36,36.35) and (123.02,26.5) .. (136.18,26.5) .. controls (149.33,26.5) and (160,36.35) .. (160,48.5) .. controls (160,60.66) and (149.33,70.51) .. (136.18,70.51) .. controls (123.02,70.51) and (112.36,60.66) .. (112.36,48.5) -- cycle ;
\draw  [color={rgb, 255:red, 0; green, 0; blue, 0 }  ,draw opacity=1 ] (112.36,104) .. controls (112.36,91.84) and (123.02,81.99) .. (136.18,81.99) .. controls (149.33,81.99) and (160,91.84) .. (160,104) .. controls (160,116.15) and (149.33,126) .. (136.18,126) .. controls (123.02,126) and (112.36,116.15) .. (112.36,104) -- cycle ;
\draw    (77.5,76.25) -- (112.36,48.5) ;
\draw    (77.5,76.25) -- (113.39,104) ;

\draw (128.18,41) node [anchor=north west][inner sep=0.75pt]  [font=\small,color={rgb, 255:red, 0; green, 0; blue, 0 }  ,opacity=1 ]  {$\Omega _{1} \ $};
\draw (128.18,97.5) node [anchor=north west][inner sep=0.75pt]  [font=\small,color={rgb, 255:red, 0; green, 0; blue, 0 }  ,opacity=1 ]  {$\Omega _{2} \ $};
\draw (45.5,71.25) node [anchor=north west][inner sep=0.75pt]  [font=\small,color={rgb, 255:red, 0; green, 0; blue, 0 }  ,opacity=1 ]  {$\Omega _{0} \ $};

\end{tikzpicture} &  \tikzset{every picture/.style={line width=0.75pt}} 

\begin{tikzpicture}[x=0.75pt,y=0.75pt,yscale=-1,xscale=1]

\draw  [color={rgb, 255:red, 0; green, 0; blue, 0 }  ,draw opacity=1 ] (215,82.25) .. controls (215,69.55) and (225.3,59.25) .. (238,59.25) .. controls (250.7,59.25) and (261,69.55) .. (261,82.25) .. controls (261,94.95) and (250.7,105.25) .. (238,105.25) .. controls (225.3,105.25) and (215,94.95) .. (215,82.25) -- cycle ;
\draw  [color={rgb, 255:red, 0; green, 0; blue, 0 }  ,draw opacity=1 ] (296.86,54.5) .. controls (296.86,42.35) and (307.52,32.5) .. (320.68,32.5) .. controls (333.83,32.5) and (344.5,42.35) .. (344.5,54.5) .. controls (344.5,66.66) and (333.83,76.51) .. (320.68,76.51) .. controls (307.52,76.51) and (296.86,66.66) .. (296.86,54.5) -- cycle ;
\draw  [color={rgb, 255:red, 0; green, 0; blue, 0 }  ,draw opacity=1 ] (296.86,110) .. controls (296.86,97.84) and (307.52,87.99) .. (320.68,87.99) .. controls (333.83,87.99) and (344.5,97.84) .. (344.5,110) .. controls (344.5,122.15) and (333.83,132) .. (320.68,132) .. controls (307.52,132) and (296.86,122.15) .. (296.86,110) -- cycle ;
\draw    (262,82.25) -- (296.86,54.5) ;
\draw    (262,82.25) -- (297.89,110) ;
\draw    (320.68,87.99) -- (320.68,76.51) ;

\draw (312.68,47) node [anchor=north west][inner sep=0.75pt]  [font=\small,color={rgb, 255:red, 0; green, 0; blue, 0 }  ,opacity=1 ]  {$\Omega _{1} \ $};
\draw (312.68,103.5) node [anchor=north west][inner sep=0.75pt]  [font=\small,color={rgb, 255:red, 0; green, 0; blue, 0 }  ,opacity=1 ]  {$\Omega _{2} \ $};
\draw (230,77.25) node [anchor=north west][inner sep=0.75pt]  [font=\small,color={rgb, 255:red, 0; green, 0; blue, 0 }  ,opacity=1 ]  {$\Omega _{0} \ $};

\end{tikzpicture} & \tikzset{every picture/.style={line width=0.75pt}} 

\begin{tikzpicture}[x=0.75pt,y=0.75pt,yscale=-1,xscale=1]

\draw  [color={rgb, 255:red, 0; green, 0; blue, 0 }  ,draw opacity=1 ] (386,178.75) .. controls (386,166.05) and (396.3,155.75) .. (409,155.75) .. controls (421.7,155.75) and (432,166.05) .. (432,178.75) .. controls (432,191.45) and (421.7,201.75) .. (409,201.75) .. controls (396.3,201.75) and (386,191.45) .. (386,178.75) -- cycle ;
\draw  [color={rgb, 255:red, 0; green, 0; blue, 0 }  ,draw opacity=1 ] (468.11,91.5) .. controls (468.11,79.35) and (478.77,69.5) .. (491.93,69.5) .. controls (505.08,69.5) and (515.75,79.35) .. (515.75,91.5) .. controls (515.75,103.66) and (505.08,113.51) .. (491.93,113.51) .. controls (478.77,113.51) and (468.11,103.66) .. (468.11,91.5) -- cycle ;
\draw  [color={rgb, 255:red, 0; green, 0; blue, 0 }  ,draw opacity=1 ] (468.61,148.75) .. controls (468.61,136.6) and (479.27,126.75) .. (492.43,126.75) .. controls (505.58,126.75) and (516.25,136.6) .. (516.25,148.75) .. controls (516.25,160.9) and (505.58,170.75) .. (492.43,170.75) .. controls (479.27,170.75) and (468.61,160.9) .. (468.61,148.75) -- cycle ;
\draw  [color={rgb, 255:red, 0; green, 0; blue, 0 }  ,draw opacity=1 ] (468.61,206) .. controls (468.61,193.84) and (479.27,183.99) .. (492.43,183.99) .. controls (505.58,183.99) and (516.25,193.84) .. (516.25,206) .. controls (516.25,218.15) and (505.58,228) .. (492.43,228) .. controls (479.27,228) and (468.61,218.15) .. (468.61,206) -- cycle ;
\draw  [color={rgb, 255:red, 0; green, 0; blue, 0 }  ,draw opacity=1 ] (468.61,266) .. controls (468.61,253.84) and (479.27,243.99) .. (492.43,243.99) .. controls (505.58,243.99) and (516.25,253.84) .. (516.25,266) .. controls (516.25,278.15) and (505.58,288) .. (492.43,288) .. controls (479.27,288) and (468.61,278.15) .. (468.61,266) -- cycle ;
\draw    (492.18,113.51) -- (492.43,126.75) ;
\draw    (492.18,170.76) -- (492.18,183.99) ;
\draw    (432,178.75) -- (468.36,91.5) ;
\draw    (432,178.75) -- (468.36,148.75) ;
\draw    (432,178.75) -- (468.61,206) ;
\draw    (492.93,228) -- (492.93,243.99) ;
\draw    (432,178.75) -- (468.61,266) ;

\draw (483.93,259) node [anchor=north west][inner sep=0.75pt]  [font=\small,color={rgb, 255:red, 0; green, 0; blue, 0 }  ,opacity=1 ]  {$\Omega _{4} \ $};
\draw (400.5,172.75) node [anchor=north west][inner sep=0.75pt]  [font=\small,color={rgb, 255:red, 0; green, 0; blue, 0 }  ,opacity=1 ]  {$\Omega _{0} \ $};
\draw (483.43,84.5) node [anchor=north west][inner sep=0.75pt]  [font=\small,color={rgb, 255:red, 0; green, 0; blue, 0 }  ,opacity=1 ]  {$\Omega _{1} \ $};
\draw (483.93,141.75) node [anchor=north west][inner sep=0.75pt]  [font=\small,color={rgb, 255:red, 0; green, 0; blue, 0 }  ,opacity=1 ]  {$\Omega _{2} \ $};
\draw (483.93,199) node [anchor=north west][inner sep=0.75pt]  [font=\small,color={rgb, 255:red, 0; green, 0; blue, 0 }  ,opacity=1 ]  {$\Omega _{3} \ $};

\end{tikzpicture} \\ \hline

Subdomain boundaries & $\Gamma_0 = \texttt{[01,02]}$  &  $\Gamma_0 = \texttt{[01,02]}$ & $\Gamma_0 =\texttt{[01,02,03,04]}$ \\ 
& $\Gamma_1 =\texttt{[-01]}$  & $\Gamma_1 = \texttt{[-01,12]}$  & $\Gamma_1 =\texttt{[-01,12,14]}$  \\
& $\Gamma_2 =\texttt{[-02]}$       & $\Gamma_2 = \texttt{[-02,-12]}$  & $\Gamma_2 =\texttt{[-02,-12,23]}$ \\
&     &  & $\Gamma_3 =\texttt{[-03,-23,34]}$  \\
&     &  & $\Gamma_4 =\texttt{[-04,-14,-34]}$ \\ \hline 
 \end{tabular}
\caption{Domain configuration for different relevant cases. In the second row, we showcase a simplified representation for each domain. This allows us to establish the connectivity graph between subdomains in the third row. In the last row, each $\Gamma_i$ is represented as a union of interface indices, where a minus sign indicates a normal swapping (i.e.~$\texttt{swap}=-1)$ in \Cref{alg:process_segment}. Note that the grid problem (right-hand column) corresponds to the scheme shown in \Cref{fig:Skeleton}.}
\label{tab:DomainConfiguration}
\end{center} 
\end{table} 
\subsection{Implementation in Bempp-cl}\label{subs:implementationbempp}
The top-down approach proposed here results in compact and intuitive routines for invoking the Calderón and transmission operators. To illustrate this further, we provide details on its implementation in Bempp-cl \cite{betcke2021bempp}. In Bempp-cl, RWG function spaces of degree \texttt{0} are typically defined on a mesh \texttt{grid} as follows:
\begin{align*}
\text{RWG:} & ~\text{\texttt{function\_space(grid, "RWG", 0)}}.
\end{align*}
This definition facilitates the assembly of operators in the form:
\be
\texttt{operator(domain, range, test)}
\ee 
with function spaces specifying the domain, range and test function spaces.  With our procedure, RWG function spaces over a subdomain $\Gamma_i$ of the skeleton mesh \texttt{grid} are defined, for $i \in \{0,\ldots,M\}$, as
\begin{align*}
\text{RWG}_i \equiv \texttt{spaces[i] }\text{:} & ~ \text{\texttt{function\_space(grid, "RWG", 0, segments, swapped\_normals)}}.
\end{align*}
Therein, \texttt{segments} represents the list of indices associated with $\Gamma_i$ and \texttt{swapped\_normals} denotes the indices where the normal is swapped, yielding the list \texttt{spaces}. For clarity, using the ``Triple point'' case (second column) in \Cref{tab:DomainConfiguration}, $\Gamma_1 = \texttt{[-01,12]}$ translates to: 
\begin{align*}
\texttt{segments} & =\texttt{[01,12]} \\
\texttt{swapped\_normals}& = \texttt{[01]}.
\end{align*} 
With the function spaces \texttt{spaces} defined, the Calderón operators can be invoked directly for $i \in \{0, \ldots, M\}$:
\be 
\texttt{multitrace\_operator(spaces[i], spaces[i], spaces[i])}.
\ee 
Consequently \eqref{eq:trick_transmission} ensures that for $i \in \{0,\ldots, M\}$ and $j\in\Lambda_i$, the transmission operator $\tmI^{ij}$ is obtained seamlessly as:
\be
\texttt{identity(spaces[j], spaces[j], spaces[i])}.
\ee 
\subsection{Preconditioning}\label{sec:Prec}
To avoid the computational cost of barycentric refinement required for traditional multiplicative Calderón preconditioning \cite{CalderonPMCHWT}, an OSRC-based preconditioner was preferred \cite{fierro2023osrc}. We introduce the electric-to-magnetic operator \cite[Eq.~8]{fierro2023osrc}:
$$
\text{EtM}_k^i : = - \left(\mT_k^i\right)^{-1} \left( \frac{\mI^i}{2} + \mK_k^i\right)
$$
which provides an approximate inverse to $\mT_k^i$, making it a suitable candidate for preconditioning for the EFIO in \eqref{eq:boundary_operators}. Its approximation on smooth surfaces is given by the Padé-approximated MtE operator:
$$
\widetilde{\mT}^i_k \equiv \tilde{\Pi}^{-1}_{k,\eps, N_p}  \approx \text{EtM}_k^i. 
$$
For any $\mu>0$, which accounts for the accuracy of the approximation, depending on $\eps$ (the damped wavenumber) and $N_p \geq 1$ (the Padé approximation order), the OSRC multiple traces operator---referred to as ``block-OSRC preconditioner'' in what follows---is defined as:
\be 
 \tbmA := \text{diag}(\tbmA_i) \quad \text{with} \quad \tbmA_i  \equiv\tbmA^i_{k}  : = \begin{bmatrix} 0  &  \varrho_i^{-1}\widetilde{\mT}^i_k  \\ -  \varrho_i \widetilde{\mT}^i_k  & 0 \end{bmatrix}, \quad i \in \{0, \ldots, M\}.
\ee
As we will see in \Cref{sec:NumExp}, the preconditioned schemes result in improved spectral properties (e.g., eigenvalues distribution, field of values, singular values), leading to moderate dependence on $h$ and grading \cite{ESCAPILINCHAUSPE2021220} for GMRES.
\section{Numerical experiments}\label{sec:NumExp}
\subsection{Methodology}
In this section, we solve EM scattering problems using Bempp-cl version 0.3.2 \cite{betcke2021bempp}. Following FAIR principles \cite{wilkinson2016fair}, all simulations are fully reproducible within a Python notebook, provided as part of the $\texttt{mtf}$ package\footnote{Refer to \url{https://github.com/betckegroup/multitrace_implementation_paper}.}. We assemble the operators using Numba dense assembly. The tests were executed on a server with 64 cores, each with 8GB RAM, powered by dual AMD EPYC 7302 processors, using Python 3.10. We introduce two cases, named A and B, corresponding to different wave frequencies, with their corresponding material parameters summarized in \Cref{table:overviewInitial}.

The incident field is defined as:
\be \label{eq:incidentNumexp}
\bE^\text{inc} = \bp e^{\imath k_0 \bd \cdot \bx}
\ee 
where $k_0 > 0$ and $\bp, \bd \in \IC^3$. The numerical experiments are structured as follows:
\begin{enumerate}
    \item In \Cref{subsec:Sphere}, we study and verify the convergence of the local MTF for a unit sphere, comparing the STF to local MTF for $M=1$ and $M=2$ (i.e.~scattering by two half-spheres);
\item In \Cref{subsec:precandperf}, we analyze the performance of OSRC preconditionign on GMRES for two half-spheres and two half-cubes;
\item In \Cref{subsec:complex}, we examine complex objects and study how the scheme behaves for increasing values of $M$.
\end{enumerate}

\begin{table}[t]
\renewcommand\arraystretch{1.2}
\begin{center}
\footnotesize
\begin{tabular}{|c|c|c|c|c|} 
\hline
\multicolumn{5}{|c|}{\bf Parameter Values}\\
\hhline{|=|=|=|=|=|}
 \multicolumn{2}{|c|}{Case}&  A& B  \\ \hline
\multicolumn{2}{|c|}{$\upmu_r$} & $1.0$ & $1.0$  \\ \hline  
\multicolumn{2}{|c|}{$\upepsilon_r$} & $2.1$  & $1.9$  \\ \hline 
\multicolumn{2}{|c|}{$f$ MHz} & $143.1$  & $238.6$ \\ \hline
 \multicolumn{2}{|c|}{$\lambda$ (m)} & $2.1$ & $1.3$ \\ \hline 
 \multicolumn{2}{|c|}{$k_0$} & $3.0$ & $5.0$  \\ \hline 
 \multicolumn{2}{|c|}{$k_1$} & $4.3$ & $6.9$ \\ \hline 
\end{tabular}
\hspace{2.0cm} 
\begin{tabular}{|c|c|c|c|c|}    
\hline
\multicolumn{5}{|c|}{\bf Results for $\texttt{MTF(2)}$}\\
\hhline{|=|=|=|=|=|}
 \multicolumn{2}{|c|}{Case}&  A& B  \\ \hline
 \multicolumn{2}{|c|}{$N$} & $6360$ & $15900$  \\ \hline  
\multicolumn{2}{|c|}{$[\text{RCS}_z]_{L^2([0,\pi])}$} & $1.23\% $ & $2.53\%$  \\ \hline  
\multicolumn{2}{|c|}{$[\gamma_D \bu]_{\bL^2(\Gamma_{12})}$} & $2.22\%$  & $1.15\%$  \\ \hline 
\multicolumn{2}{|c|}{$[\gamma_N \bu]_{\bL^2(\Gamma_{12})}$} & $2.27\%$  & $0.39\%$  \\ \hline 
\end{tabular}
\end{center}
\caption{(Left) Overview of the relative material parameters, frequency $f$, wavelength $\lambda$ and wavenumbers $k_i$, $i=0,1$ for each case. (Right) Results for $\texttt{MTF(2)}$ for $r=10$. We show the relative errors for $\text{RCS}_z$ and the Dirichlet and Neumann jumps over the interface $\Gamma_{01}$.}
\label{table:overviewInitial}
\end{table}  

\subsection{Scattering by a sphere}\label{subsec:Sphere}First, we consider the scattering of a unit sphere for cases A and B in \Cref{table:overviewInitial}. We set the incident field in \eqref{eq:incidentNumexp} with
$$
 \bp : =  \begin{bmatrix} 0, 0,  1 \end{bmatrix}^T \quad \text{and}\quad \bd :=  \begin{bmatrix} 1, 0,  0 \end{bmatrix}^T .
 $$
We use the MIE series \cite{miebookabsorbtion} for the exact solution and evaluate $\text{RCS}_z$, the $z$-component of the far-field at $z=0$. The resulting linear systems are solved using direct LU decomposition (see \cite{betcke2021bempp} for more details). We compare the following methods:
\begin{enumerate}
    \item \texttt{STF(1)}: the STF (PMCHWT) for $M=1$;
    \item \texttt{MTF(1)}: the local MTF for $M=1$;
    \item \texttt{MTF(2)}: the local MTF for $M=2$ (two half-spheres).
\end{enumerate}
In \texttt{MTF(2)}, we divide the unit sphere sphere into two regions $\Omega_1 := \{ \bx = (x,y,z) : \|\bx\|_2 \leq 1 \text{ and } x <1\}$ and $\Omega_2 := \{ \bx = (x,y,z): \|\bx\|_2 <  1 \text{ and } x > 1\}$, where $\Omega_2$ inherits the material parameters from $\Omega_1$. We begin by representing the scattered and total fields for \texttt{MTF(2)} in \Cref{fig:s1-sol} for $r=10$, with case A shown at the top and case B at the bottom.
\begin{figure}[htb!]
\center
 \begin{subfigure}[b]{1\textwidth}
 \centering
\includegraphics[width=1\textwidth]{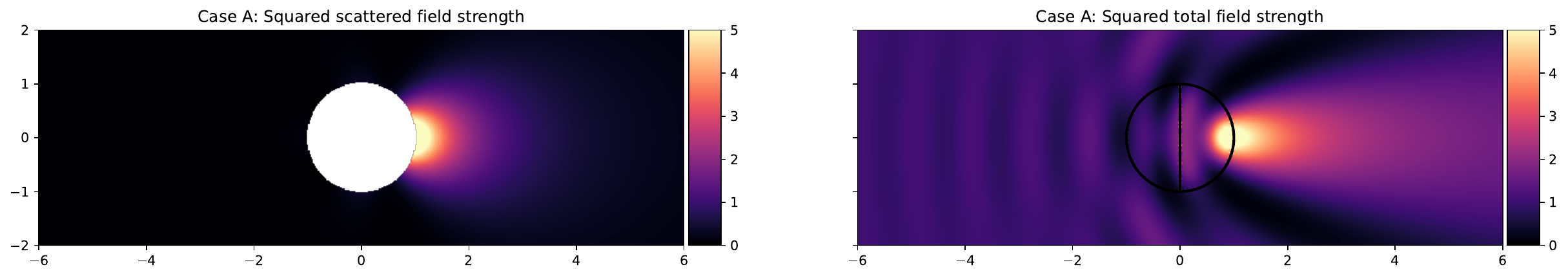}
\caption{Case A.}
\end{subfigure}
\begin{subfigure}[b]{\textwidth}
\includegraphics[width=1\textwidth]{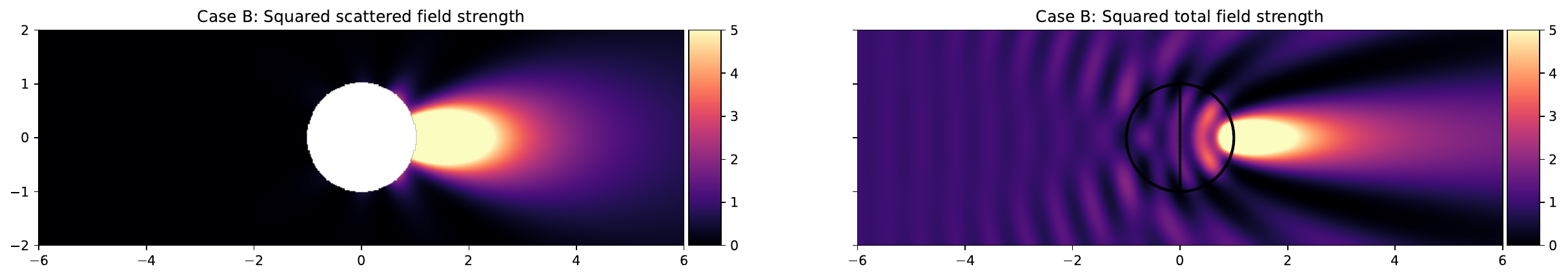}
\caption{Case B.}
\end{subfigure}
\caption{$\texttt{MTF(2)}$ for Case A (top) and Case B (bottom): Squared amplitude in plane $z=0$  for the scattered field $|\bE^\text{sc}|^2$ (left) and total field $|\bE|^2$ (right) for a precision of $r=10$ points per wavelength. }
\label{fig:s1-sol}       
\end{figure}
In \Cref{table:overviewInitial} (right), we present:
\begin{itemize}
    \item $N$: the size of the local MTF matrix;
    \item $[\text{RCS}_z]_{L^2([0,\pi])}$: the relative $L^2([0,\pi])$ error for RCS$_z$;
    \item $[\bgamma_\cdot \bu ]_{\bsL^2(\Gamma_{12}) }$, $\cdot \in \{D,N\}$ the relative $\bsL^2$-norm error for the Dirichlet and Neumann jumps over $\Gamma_{12}$.
\end{itemize}    
We observe that the errors are low for both $\text{RCS}_z$ and the jumps, ranging from $0.39\%$ to $2.53\%$. Fig.~\ref{fig:Sphere}(a) (for Case A) and ~Fig.~\ref{fig:Sphere}(d) (for Case B) display $\text{RCS}_z$.

Next, we examine $h$-convergence by generating a sequence of increasingly refined meshes $\Sigma^h$ corresponding to precisions of $r \in [1,2, 5,10, 20, 30, 40]$ points per wavelength. For case B, at higher precisions $r \in [30,40]$, the solution is obtained through GMRES, as a direct solver was not feasible. Fig.~\ref{fig:Sphere}(b) and Fig.~\ref{fig:Sphere}(d) show the convergence results for $[\text{RCS}_z]_{L^2([0,\pi])}$ in Case A and Case B, respectively. We remark that the error for all three formulations decreases as $\mO (h^2)$, with convergence starting from low resolution at $r=2$ for Case A and $r=1$ for Case B.
The top $x$-axis represents the precision $r$. Interestingly, the three formulations exhibit similar convergence patterns; the presence of the triple point in $\texttt{MTF(2)}$ does not degrade far-field convergence. Finally, Fig.~\ref{fig:Sphere}(d) and Fig.~\ref{fig:Sphere}(f) show $[\gamma_\cdot \bu ]_{\bL^2(\Gamma_{12}) }$ for $\cdot
\{D,N\}$ in $\texttt{MTF(2)}$, with a decrease $\mO(h^2)$. 

\begin{figure*}[htb!]
    \centering
    \begin{minipage}{0.32\linewidth}
\begin{figure}[H]
\includegraphics[width=0.95\linewidth]{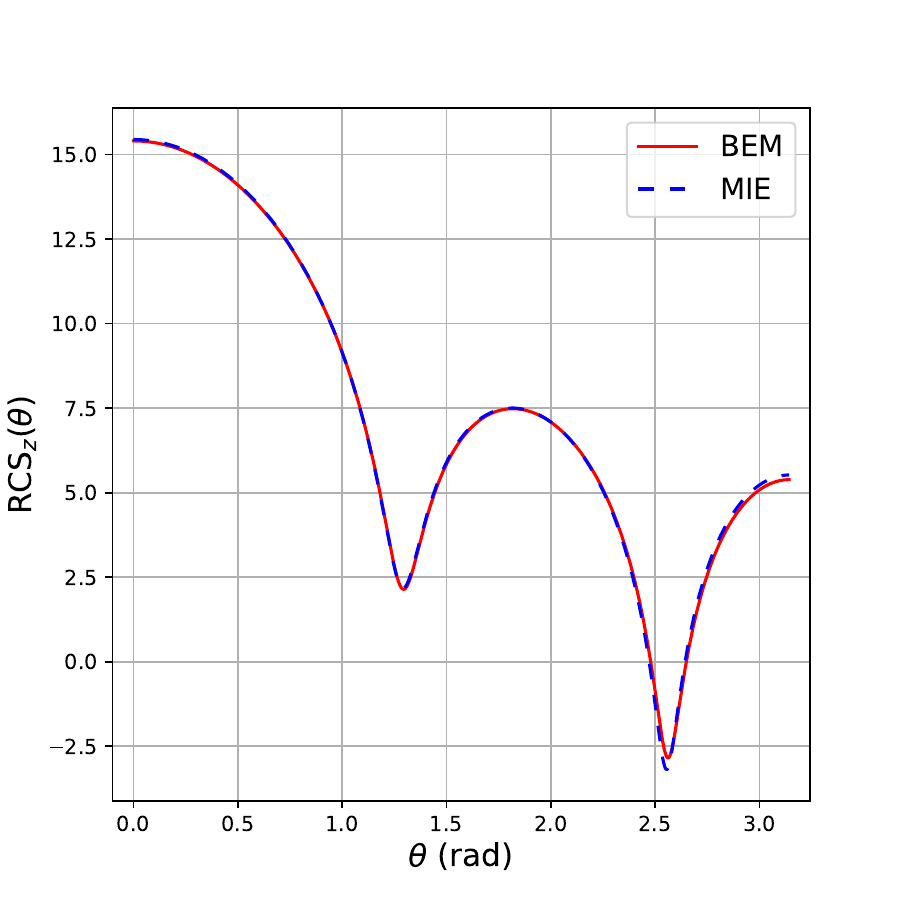}
\vspace{-0.2cm}
\caption*{(a)}
\end{figure}
\vspace{-0.9cm}
\begin{figure}[H]
\includegraphics[width=0.95\linewidth]{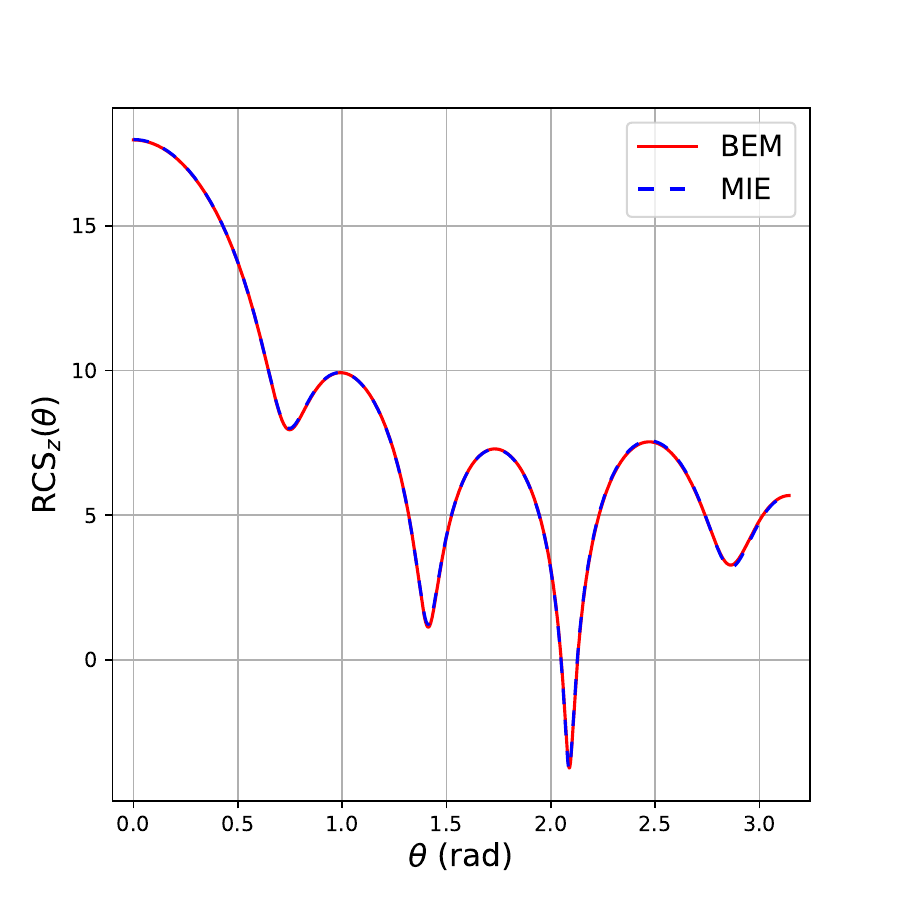}
\caption*{(d)}
\end{figure}
    \end{minipage}
    \begin{minipage}{0.32\linewidth}
    \begin{figure}[H]
\includegraphics[width=0.95\linewidth]{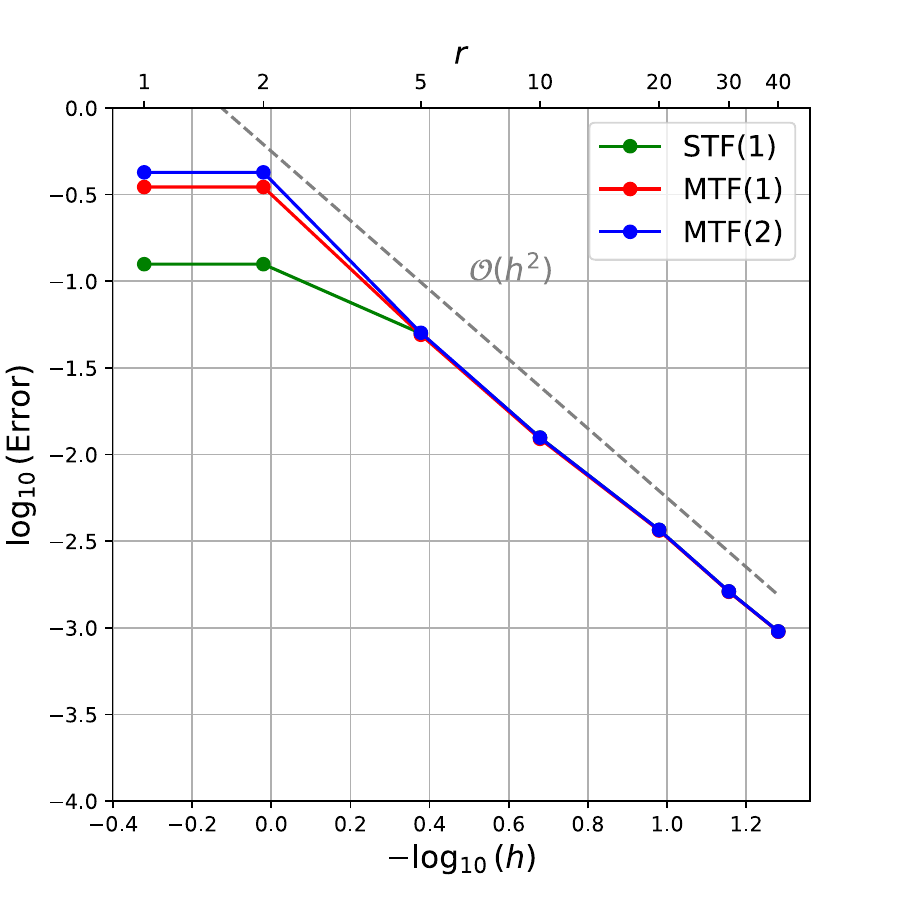}
\caption*{(b)}
\end{figure}
\vspace{-0.9cm}
\begin{figure}[H]
\includegraphics[width=0.95\linewidth]{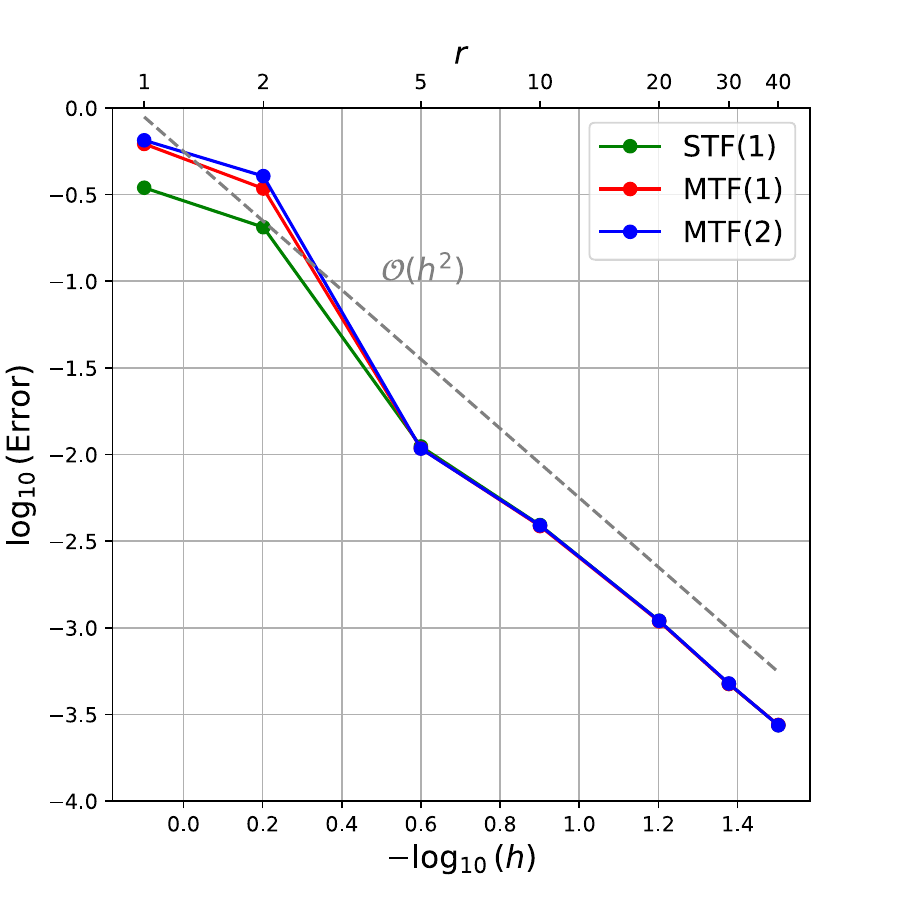}
\caption*{(e)}
\end{figure}
    \end{minipage}{}
  \begin{minipage}{0.32\linewidth}
    \begin{figure}[H]
\includegraphics[width=0.95\linewidth]{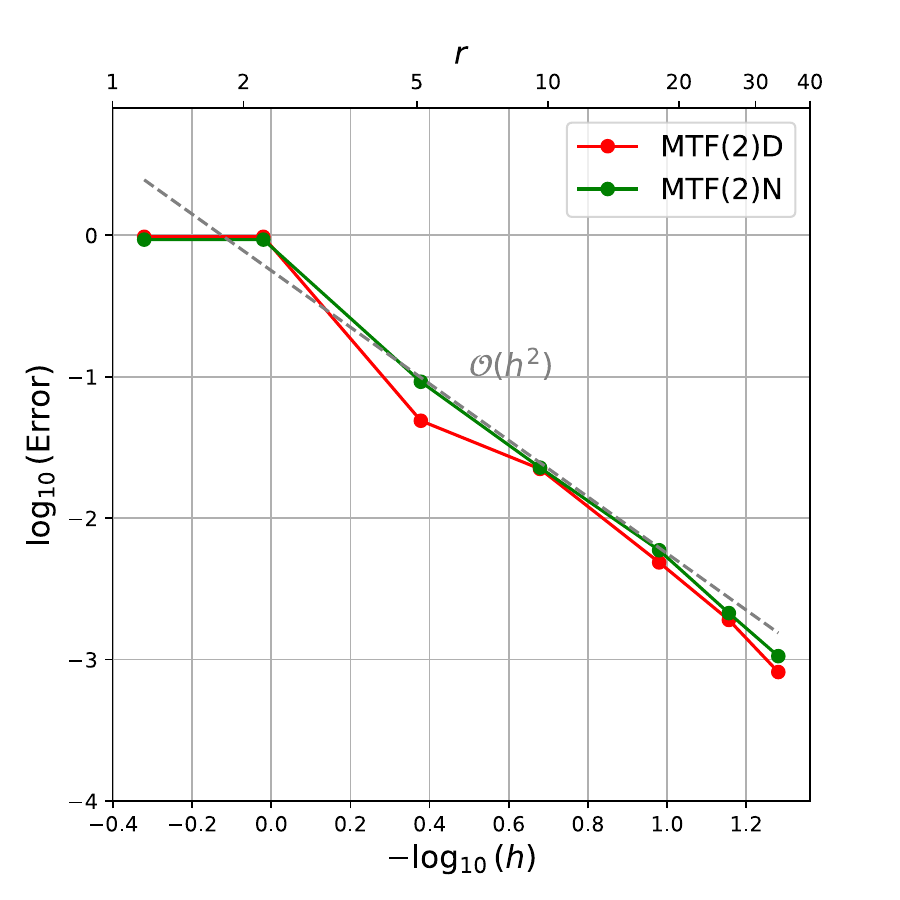}
\caption*{(c)}
\end{figure}
\vspace{-0.9cm}
\begin{figure}[H]
\includegraphics[width=0.95\linewidth]{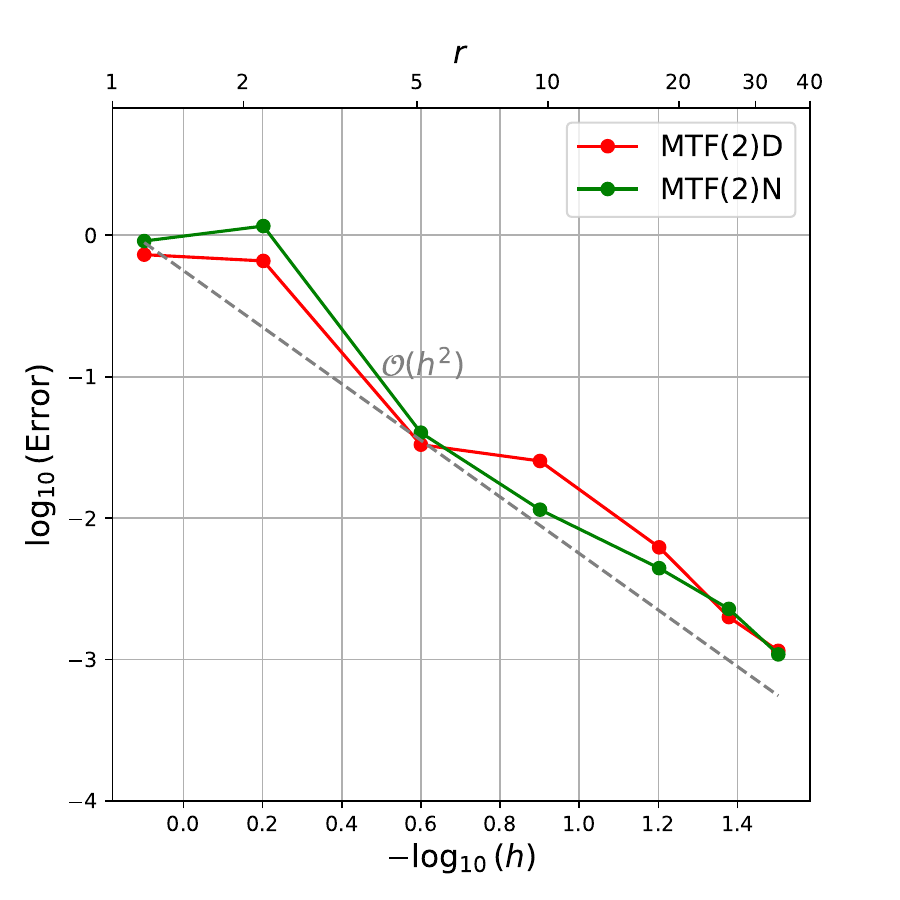}
\caption*{(f)}
\end{figure}
    \end{minipage}{}
    \caption{Case A (top) and Case B (bottom): Results for unit sphere with a known analytic solution and varying problem complexity. Left (a) (d): $\text{RCS}_z(\theta)$ for $r=10$. Middle (b) (e): Relative $L^2$ error in RCS$_z$ versus $h$. Right (c) (d): Relative $\bL^2$-error for jumps over interface $\Gamma_{12}$ versus $h$.}
    \label{fig:Sphere}
\end{figure*}

\subsection{Preconditioning and performance}\label{subsec:precandperf}
We set the incident field in \eqref{eq:incidentNumexp} for Case A and Case B in \Cref{table:overviewInitial} with
$$
 \bp : =  \begin{bmatrix} 1 + \imath ,2,  -1 - \frac{\imath}{3} \end{bmatrix}^T \quad \text{and}\quad \bd := \frac{1}{\sqrt{14}}\begin{bmatrix} 1, 2,  3 \end{bmatrix}^T 
 $$
We study GMRES for $\texttt{MTF(2)}$ in \Cref{subsec:Sphere} and extend the analysis to include scattering by a half-cube, defined as $\Omega_1 : \{ (x,y,z) \in(0,1) \times (0,1) \times (0,0.5), \}$ and $\Omega_2 := \{ (x,y,z)\in(0,1) \times (0,1) \times (0.5,1) \}$, where $\Omega_2$ inherits the material parameters from $\Omega_1$. To investigate the dependence of $h$ on GMRES convergence, we examine precisions of $r_0=10$ and $r_1=20$ points per wavelength, yielding linear systems of size $N_0$ and $N_1$, respectively. The half-sphere and half-cube meshes for both precisions are shown in \Cref{fig:meshes2}.
\begin{figure}[htb!]
\center
\includegraphics[width=\textwidth]{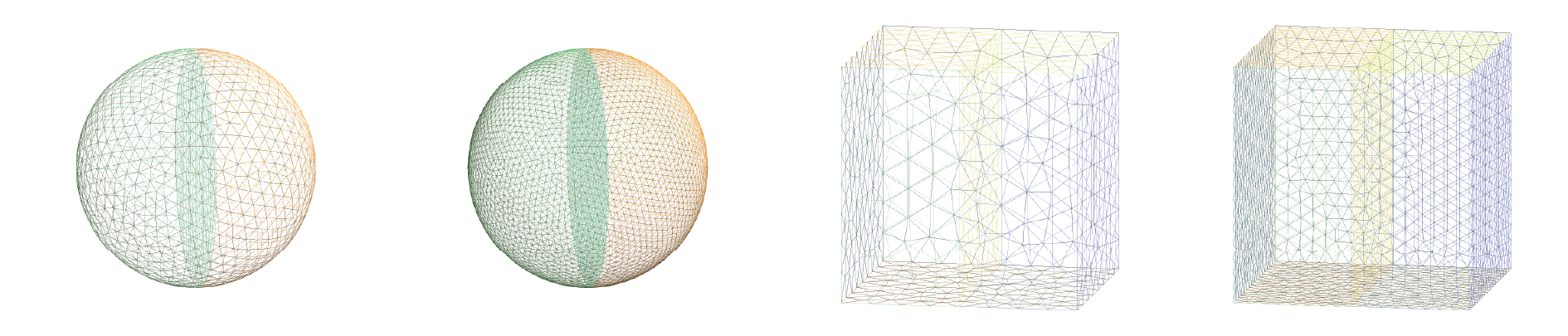}
\caption{Mesh for the half-sphere and half-cube for precisions of $r=10$ and $r=20$ points per wavelength.}
\label{fig:meshes2}       
\end{figure}
The linear system is solved using GMRES with a tolerance of $10^{-5}$ and a maximum iteration count of $2$,$000$. \Cref{fig:convGmres} presents the GMRES convergence results for the half-sphere (top) and the half-cube (bottom). 
\begin{figure}[htb!]
\center
\includegraphics[width=1\textwidth]{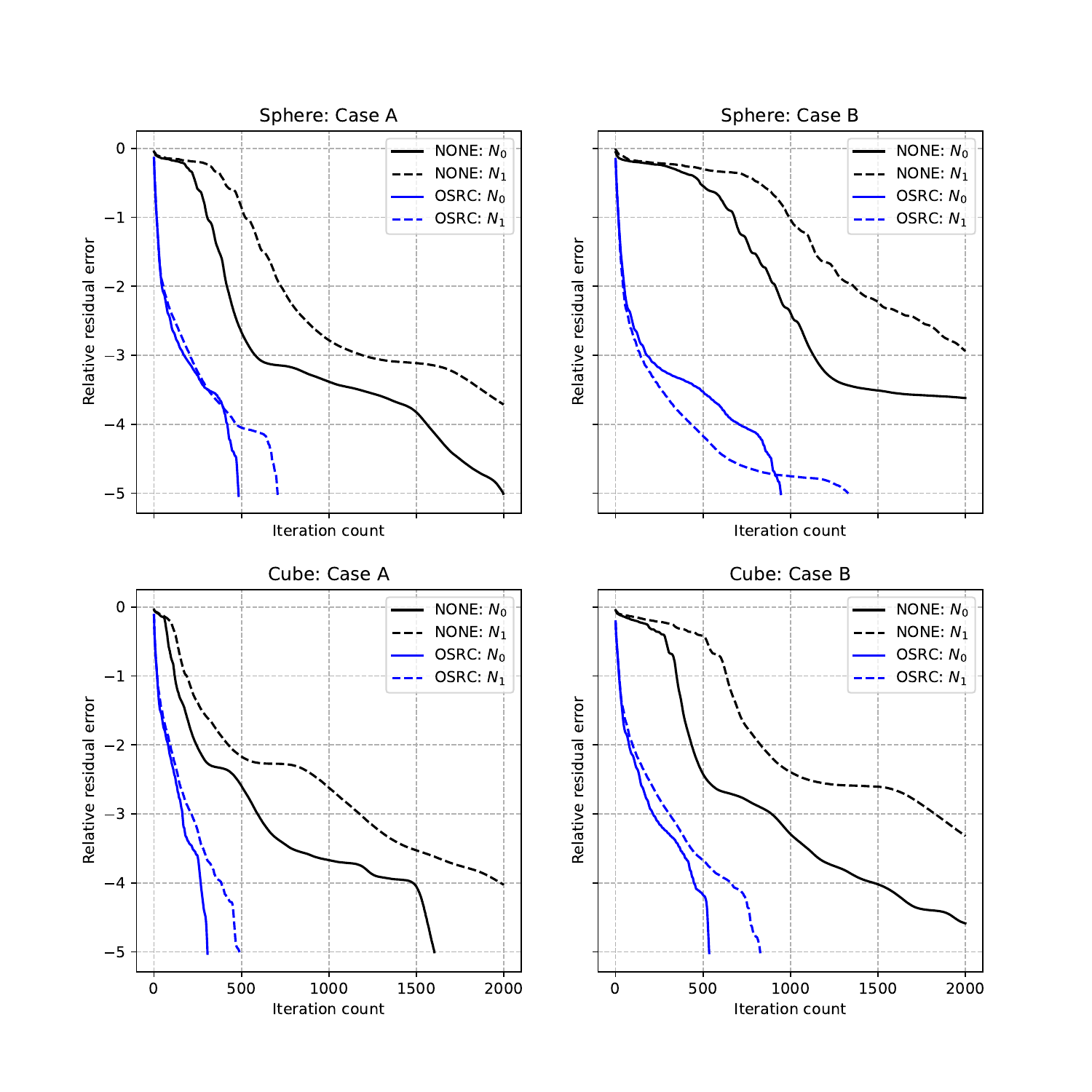}
\caption{MTF for $M=2$. Convergence of GMRES is shown for the half-sphere (top) and the half-cube (bottom) in Case A (left) and Case B (right). We compare unpreconditioned GMRES (black) with block-OSRC preconditioned GMRES (blue). Convergence results are provided for $N_0$ (solid lines) and $N_1>N_0$ (dashed lines), enabling the study of $h$-dependence, and demonstrating the robust performance of the block-OSRC preconditioned scheme.}
\label{fig:convGmres}       
\end{figure}
Each plot displays the convergence behavior for the block-OSRC preconditioner (in blue) compared to the unpreconditioned case (in black). Solid lines correspond to $r=10$, while dashed lines represent $r=20$. These plots reveal that the convergence behavior for both the half-sphere and the half-cube is similar. Block-OSRC preconditioning demonstrates robustness with respect to $h$, especially in the initial stages of convergence. In contrast, the non-preconditioned formulation shows to be more sensitive to $h$ and converges more slowly, failing in most cases to converge to the tolerance. In summary, preconditioning is indispensable. Block-OSRC preconditioning not only accelerates convergence by a factor of more than $3$, but also ensures convergence at higher frequencies as demonstrated in Case B, particularly for the sphere. This underscores the feasibility of the method and its suitability as a robust computational framework. 

Although providing only partial information on GMRES \cite{embree2022descriptive}, we examine the distribution of eigenvalues for $r=20$. Eigenvalues for the block-OSRC preconditioned formulation are obtained iteratively using the $\texttt{scipy.sparse.linalg.eigs}$ routine based on Arnoldi iteration. \Cref{fig:eigvals} shows the eigenvalues distribution for Case A and $r=10$ for both the half-sphere (top) and the half-cube (bottom). We remark that the distribution is similar. The left-hand side of \Cref{fig:eigvals} shows the overall distribution, while the right-hand side provides a zoomed view around the origin. We remark that, for the OSRC, the eigenvalues are bounded away from zero. 

\begin{figure}[htb!]
\center
\includegraphics[width=1\textwidth]{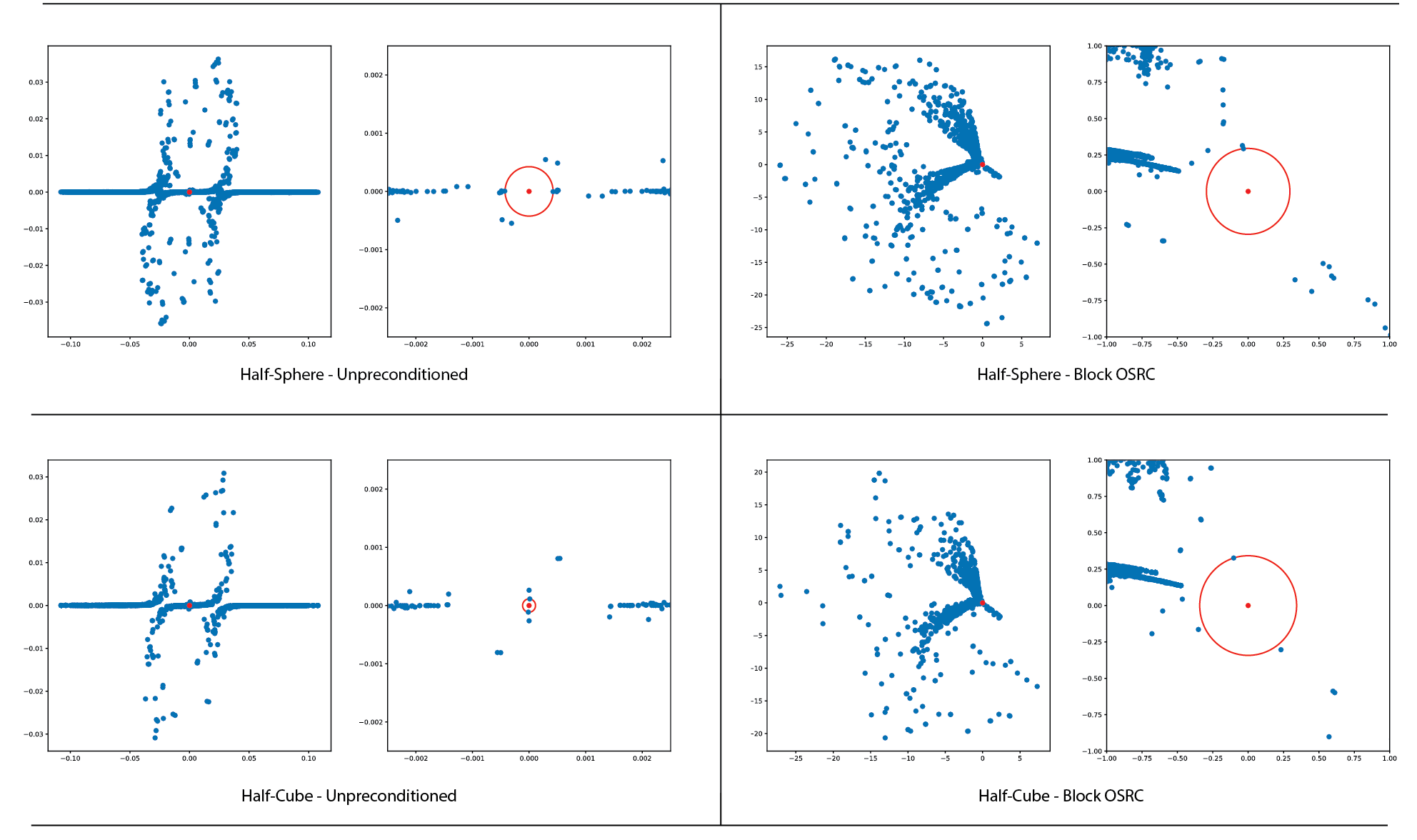}
\caption{MTF for $M=2$. Eigenvalues distribution for the local MTF (left) and block-OSRC preconditioned system (right).  We show results for the half-sphere (top) and the half-cube (bottom).}
\label{fig:eigvals}       
\end{figure}

\subsection{Scattering by multiple concentric cuboids}\label{subsec:complex}
To complete the numerical experiments\footnote{The script for this problem was  included as a tutorial in bempp-cl.}, we introduce a composite object consisting in a partition into concentric, nested cuboids. To this end, we define a bounded scatterer $\Omega_S := \{ (x,y,z) \in (-1.7,1.7) \times (-1.7, 1.7) \times (0,1) \}$, which will be divided into $M$ concentric cuboids for $M \in \{1,\ldots,6\}$. 

For $M \in \{1, \ldots, 6\}$, we introduce the \emph{radius vector} $\brr^M = [r_0^M,\ldots,r_{M}^M]$ of size $M+1$. With this vector at hand, we define the segments:  
$$
l_i^M =(-r_{i+1}^M, -r_i^M) \cup (r_i^M, r_{i + 1}^M),
$$
with $i \in \{0,\ldots,M\}$. Each concentric cuboid is then set as:
$$
\Omega_i : =  \big{\{} (x,y,z) \in  l_i^M \times l_i^M \times (0,1) \big{\}} .
$$
\Cref{fig:complexCaseMeshes} shows the meshes for $M=2,4,6$, with $\brr$ highlighted in dark blue. For instance, for $M=2$ (left subfigure), $\brr^2 = [0,0.5, 1.7]$. The definition of each radius vector $\brr^M$ for $M \in \{1, \ldots, 6\} $ is provided in Fig.~\ref{fig:ComplexOverview}(a).

\begin{figure}[htb!]
    \centering
    \begin{minipage}{0.32\linewidth}
\begin{figure}[H]
\includegraphics[width=1\linewidth]{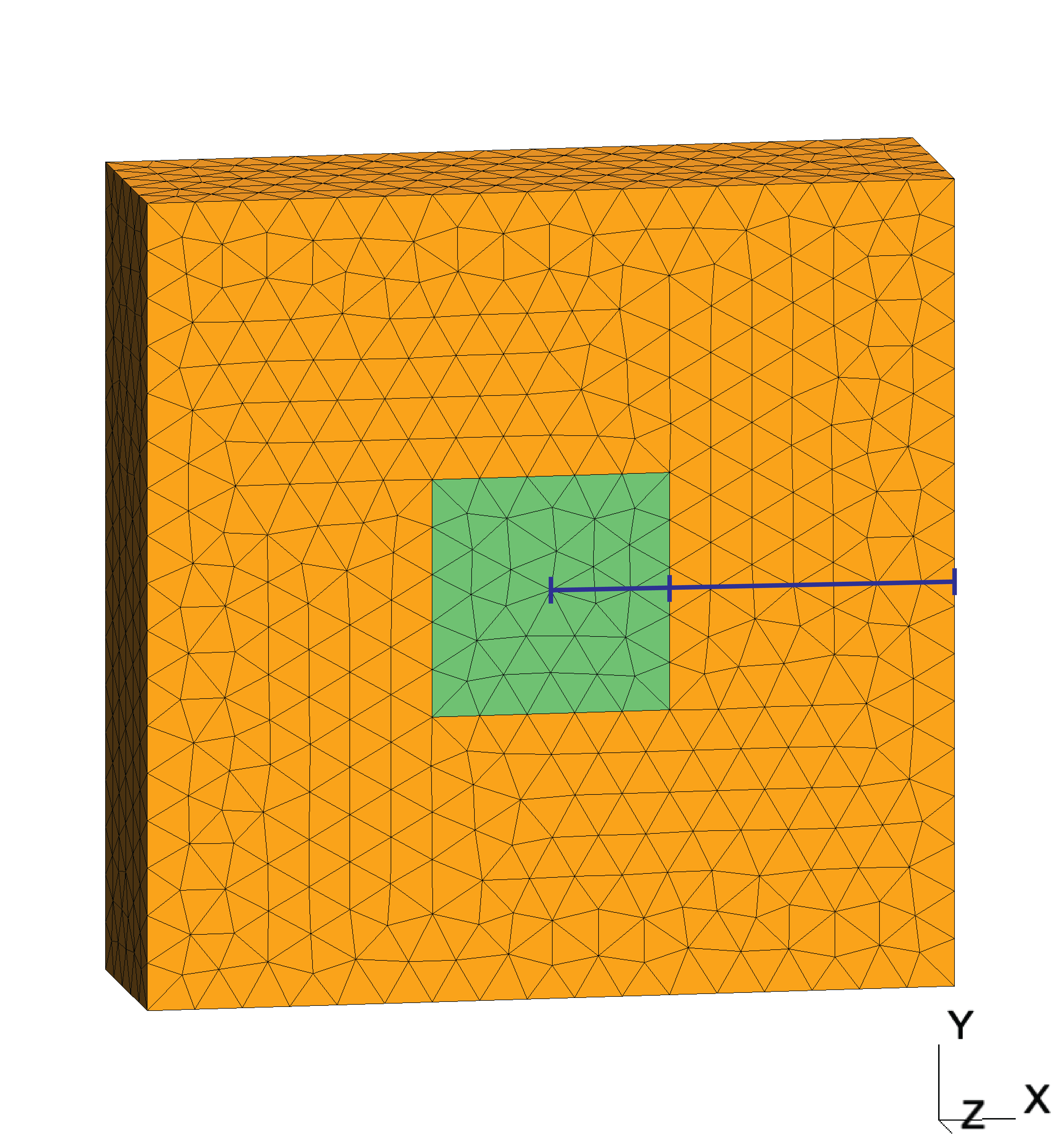}
\caption*{Mesh for $M=2$.}
\end{figure}
    \end{minipage}
    \begin{minipage}{0.32\linewidth}
    \begin{figure}[H]
\includegraphics[width=1\linewidth]{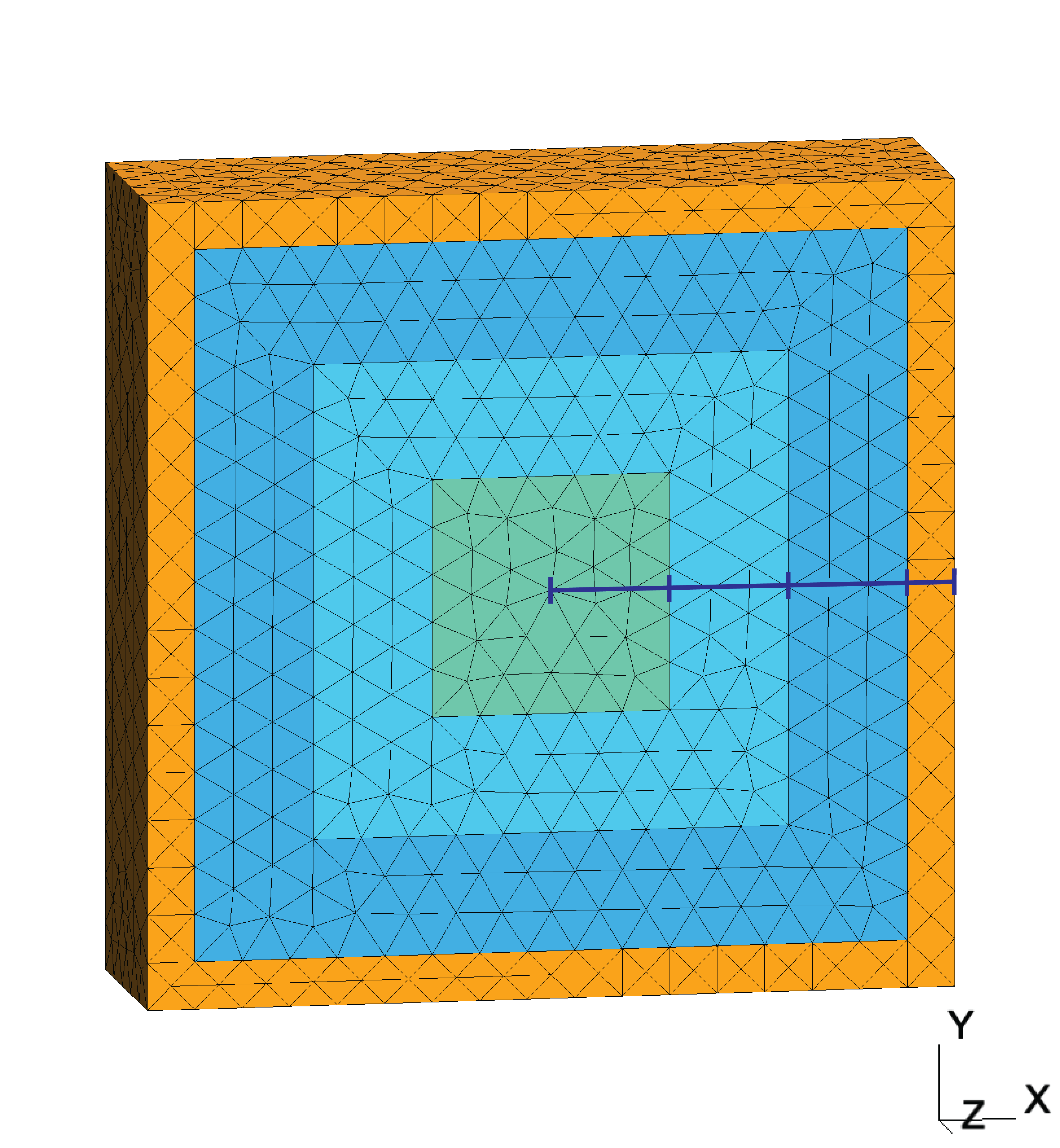}
\caption*{Mesh for $M=4$.}
\end{figure}
    \end{minipage}{}
    \begin{minipage}{0.32\linewidth}
    \begin{figure}[H]
\includegraphics[width=1\linewidth]{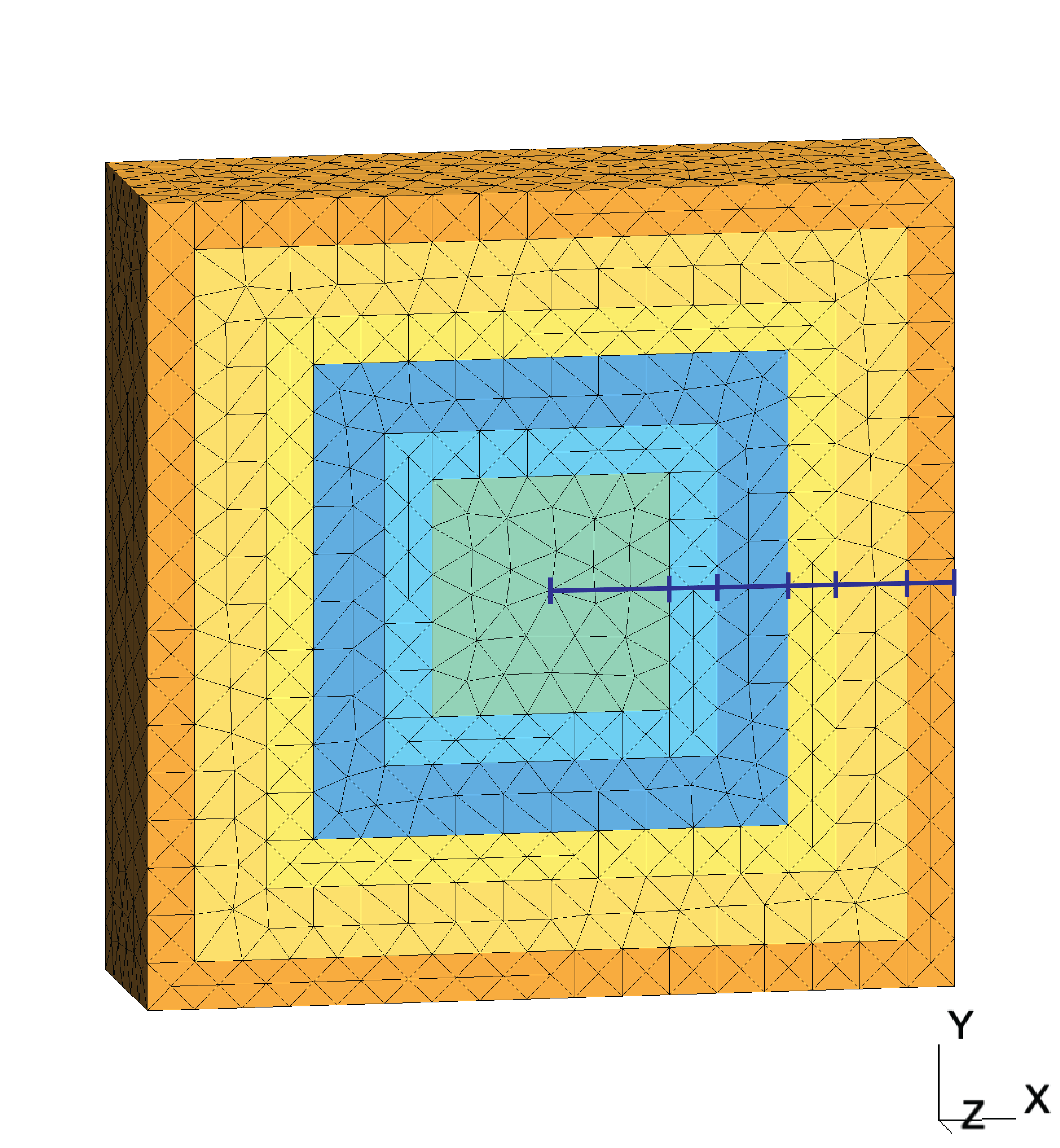}
\caption*{Mesh for $M=6$.}
\end{figure}
    \end{minipage}{}
\caption{Multiple concentric cuboids. Meshes for $M=2,4,6$ corresponding to a precision of $r=10$. For clarity, the segments representing the coordinates for $\brr^M$ are shown in dark blue. Refer also to Fig.~\ref{fig:ComplexOverview}(a).}
\label{fig:complexCaseMeshes}
\end{figure}

In order to study the local MTF as a domain decomposition method, we give the same material properties to all $\Omega_i$, for $i=1, \ldots, M$, and compare solutions for increasing $M$. Specifically, we employ the parameters of Case A: $k_0=3.0$ and $k_i=4.3$ for $i=1,\ldots,M$. Furthermore, we set $\eps_{r,i} = 2$ and $\mu_{r,i}=1$ for $i\in \{1,\ldots,M\}$. We apply block-OSRC preconditioning of Type 2 \cite[Section 3.5]{fierro2023osrc}, which demonstrated similar convergence to standard preconditioning (of Type 1) for a twofold decrease in memory and solver times approximately, in line with the bi-parametric operator preconditioning principles, which state that GMRES is provably robust to compression \cite{ESCAPILINCHAUSPE2021220}. We use a reference solution obtained by applying the MTF with $M=1$, a precision of $r=30$, and a GMRES tolerance of $10^{-5}$, resulting in a linear system with $N=108$,$432$ that converges after 365 iterations only when preconditioned. Indeed, the unpreconditioned formulation does not converge and yields a relative residual of only $0.63$ at step 365.

\begin{figure}[htb!]
    \centering
    \begin{minipage}{0.49\linewidth}
\begin{figure}[H]
\includegraphics[width=1\linewidth]{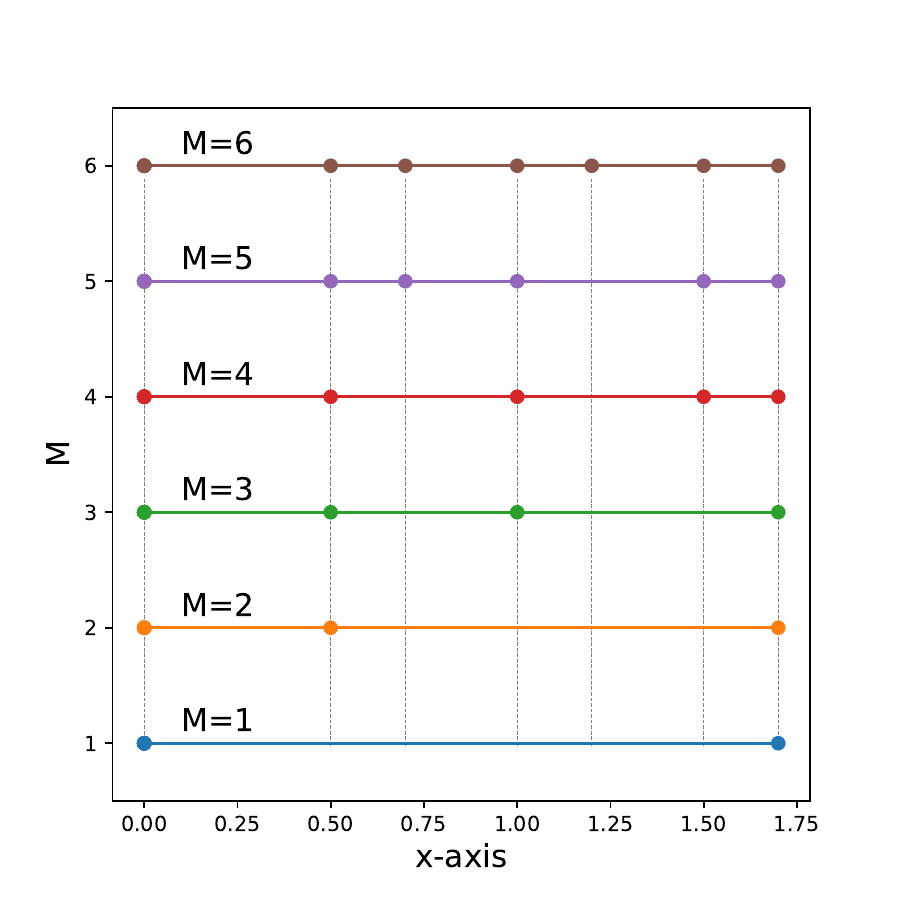}
\caption*{(a)}
\end{figure}
\begin{figure}[H]
\includegraphics[width=1\linewidth]{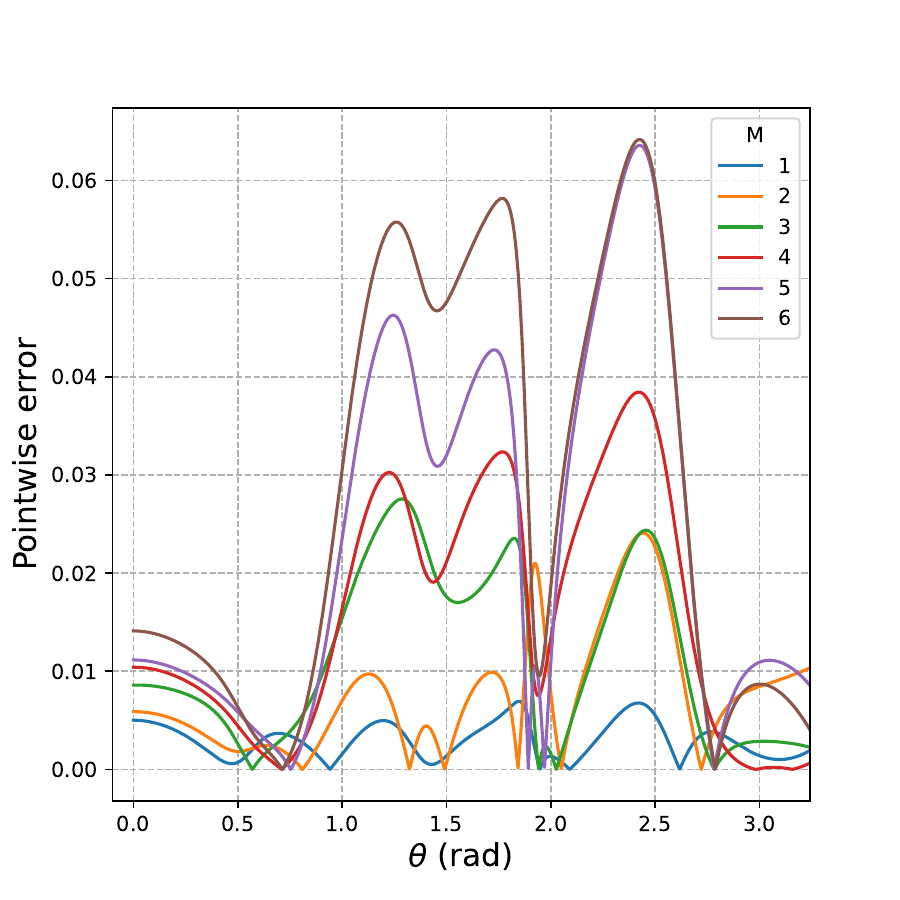}
\caption*{(c)}
\end{figure}
    \end{minipage}
    \begin{minipage}{0.49\linewidth}
\begin{figure}[H]
\includegraphics[width=1\linewidth]{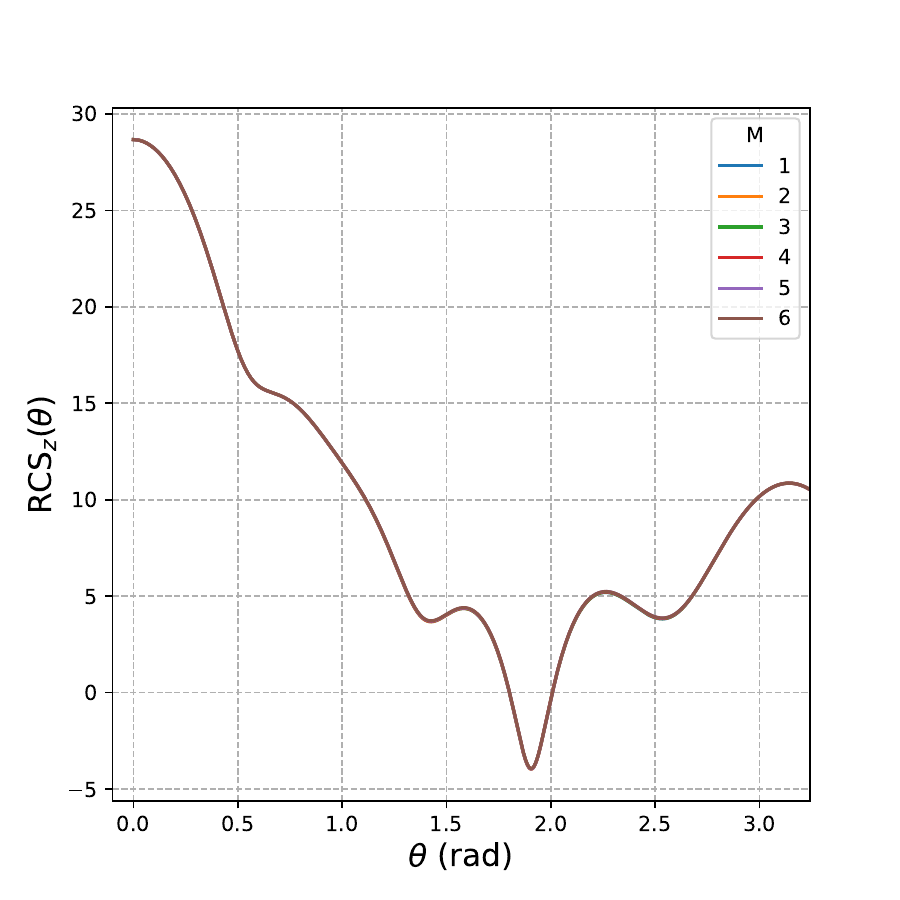}
\caption*{(b)}
\end{figure}
\begin{figure}[H]
\includegraphics[width=1\linewidth]{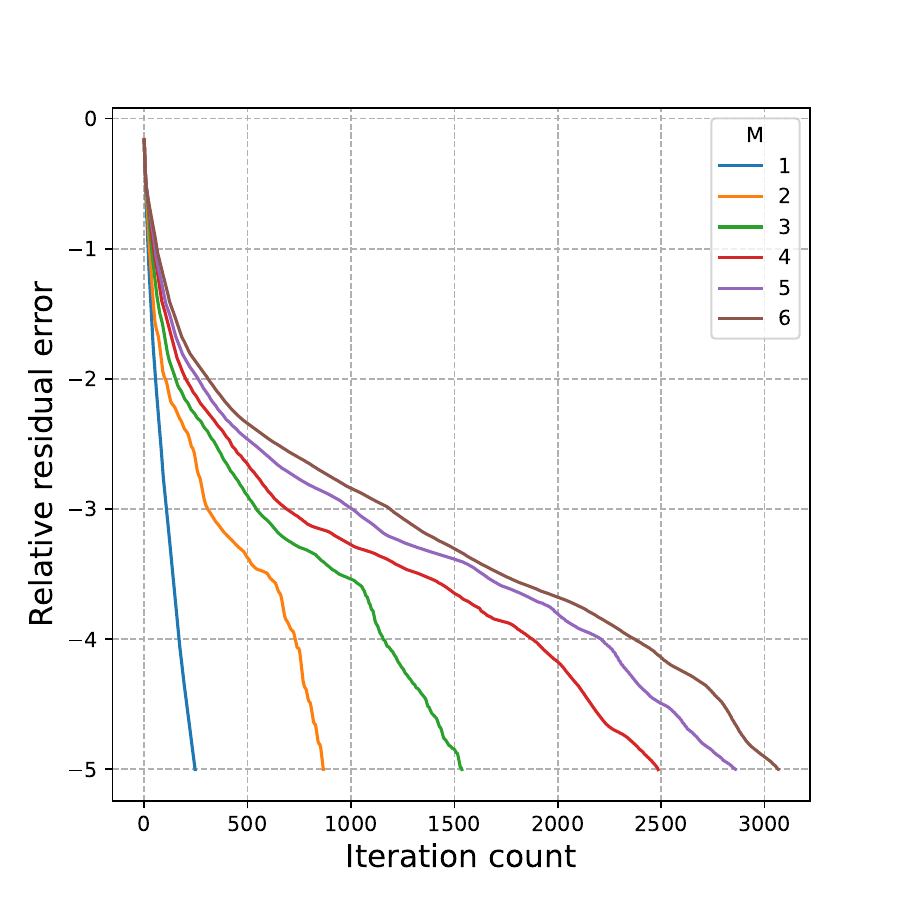}
\caption*{(d)}
\end{figure}
    \end{minipage}
\caption{Multiple concentric cuboids. Overview of the results for $M\in \{1,\ldots,6\}$. (a) Description of $\brr^M$. (b) RCS$_z(\theta)$. (c) Absolute difference between RCS$_z(\theta)$ and the reference solution. (d) Convergence results for GMRES.}
\label{fig:ComplexOverview}
\end{figure}

\Cref{fig:ComplexOverview} and \Cref{tab:summaryFinal} provide an overview of the results for $M \in \{1, \ldots, 6\}$. In  Fig.~\ref{fig:ComplexOverview}(a), we recall the expression for $\brr^M$. Fig.~\ref{fig:ComplexOverview}(b) plots RCS$_z(\theta)$, showing that the solutions remain highly similar as $M$ increases, with errors reported in the last column of \Cref{tab:summaryFinal}. Fig.~\ref{fig:ComplexOverview}(c) depicts the absolute difference between RCS$_z(\theta)$ and the reference solution for all values of $M$, demonstrating that the error remains small. This error grows at a rate of $\mathcal{O}(M^{1.30})$ and sub-linearly at $\mathcal{O}(M^{0.82})$ when normalized by the number of dofs. These findings align with the bounds found for disjoint scatterers \cite{HJH18,MJP24}. We argue that since changes in the scatterer topology affect the number of interfaces, a fair comparison to assess the impact of increasing the number of subdomains should be performed for a normalized number of dofs. Finally, Fig.~\ref{fig:ComplexOverview}(d) illustrates the GMRES convergence behavior, showing an increase in the final iteration count as $M$ grows.

\begin{table}[htb!]
\renewcommand\arraystretch{1.4}
\begin{center}
\footnotesize

\begin{tabular}{|c|c|c|c|c|c|c|c|c|c|c|c|} \hline
   \multicolumn{2}{|c|}{Case} & GMRES & \multicolumn{5}{c|}{Execution times (s)} & \multicolumn{1}{c|}{Memory requirements}& Error \\ \hline
 $M$&  $N$&    $n_\text{GMRES}$   & $t^\text{solver}$& $t^\text{assembly}_{\bP}$&   $t^\text{assembly}_{\bM}$ &  $t^\text{matvec}_{\bP}$ &$t^\text{matvec}_{\bM}$& $\text{nnz}_{\bM} \times 10^{-6}$  &$[\text{RCS}_z]_{L^2([0,\pi])}$\\ \hline\hline 
 1 & 13,440& 247 & 22.9 & 4.1 & 7.2 &0.78 & 7.3 & 205.81 & 1.57e-04 \\ \hline
 2 & 15,120&867 & 105.2 & 5.1 & 7.7 &1.1 & 7.9 &250.48 & 4.32 e-04 \\ \hline
 3 & 18,252 &1,537 & 244.1 & 18.2 & 7.4 & 1.4 & 8.6 & 256.55 & 6.52 e-04 \\ \hline 
 4 &  24,012 &2,486 & 601.4 & 19.3 & 8.0 & 2.0 & 11.5 &  322.30&  9.12e-04  \\ \hline 
 5 & 26,736   &2,860& 806.7 & 9.6 &  8.1 &  2.6  & 12.6&  339.41& 1.37e-03\\ \hline
 6 &  31,380        & 3,067 & 1,052.3      &   35.3    & 8.4&  3.0   & 15.6  &373.33 & 1.62e-03\\ \hline\hline
  eoc& $M^{0.48}$&$M^{1.44}$&$M^{2.20}$&$M^{1.02}$ & $M^{1.10}$&$M^{0.44}$ & $M^{0.41}$& $M^{0.33}$& $M^{1.30}$\\ \hline \hline
    s-eoc & - &$M^{0.96}$&$M^{1.72}$&$M^{0.54}$ & $M^{0.62}$&$M^{-0.04}$ & $M^{-0.07}$& $M^{-0.15}$& $M^{0.82}$\\ \hline 
\end{tabular}
\end{center} 

\caption{Multiple concentric cuboids. Numerical results for growing numbers of $M$. The bottom rows provide rates with respect to the number of subdomains with and without normalizing by total degrees of freedom.}
\label{tab:summaryFinal} 
\end{table}  

As mentioned, the illustrations in \Cref{fig:ComplexOverview} are complemented by \Cref{tab:summaryFinal}, which presents various metrics for each $M$. The second to last row highlights the expected order of convergence (eoc) for each metric with respect to $M$, represented as the best interpolation by a polynomial of the form $C M^\alpha$, where $C$ is a constant and $\alpha$ is the eoc. Likewise, the last row shows the \emph{standardized} eoc (s-eoc) for each metric, referred to as s-form, and scaled by $N$, to study it for the same number of dofs and provide fair comparisons, as explained before. We note that $N$, the size of the local MTF, grows approximately as $\sqrt{M}$, ranging between 13,440 and 31,380 dofs. In addition, we report $n_\text{GMRES}$, the number of iterations for GMRES to reach the relative tolerance of $10^{-5}$. The iteration count increases approximately as $M^{1.44}$ (resp.~$M^{0.96}$ in s-form), ranging between 247 and 2,067 iterations. Although convergence slows with increasing $M$, it is noteworthy that the method achieves convergence to a low tolerance for six subdomains. The growth in iteration numbers directly impacts the solver times, $t^\text{solver}$, which ranges between 22.9 seconds and 1,052.2 seconds (approximately 18 minutes). The solver times grow as $M^{2.20}$ (resp.~$M^{1.72}$ in s-form), which is close to quadratic growth. However, the assembly times for both the preconditioner $t^\text{assembly}_{\bP}$ and the local MTF matrix $t^\text{assembly}_{\bM}$ grow approximately linearly with $M$. The time required to perform $100$ matrix-vector products, referred to as $t^\text{matvec}_\star$, $\star \in \bP, \bM$ increases slowly with $M$. The memory and time requirements for the preconditioner are relatively low compared to those for $\bM$. Finally, the memory requirements for $\bM$ grow as $M^{0.33}$.

Based on the above observations, one can state the following key insights:
\begin{enumerate}
    \item The local MTF successfully solves increasingly complex problems, namely two half-spheres, two half-cubes, and multiple concentric cuboids, within a domain decomposition framework while maintaining controlled error; 
    \item Solver times are the limiting factor, as they outweigh both assembly times and memory requirements;
    \item While both Type 1 and Type 2 block-OSRC preconditioners were relatively inexpensive, there is potential to explore trade-offs between cost and efficiency by investigating more preconditioners. 
    \end{enumerate}
    
    According to bi-parametric operator preconditioning theory, a strong continuous preconditioner performs well and offers opportunities for significant compression \cite{ESCAPILINCHAUSPE2021220}. In this case, a denser continuous OSRC could be explored, though it does not guarantee $M$-independent iteration counts. In contrast, Calderón preconditioning typically yields lower iteration counts, but can be prohibitively costly. Future work should focus on achieving a compromise between these approaches.

\section{Conclusions}
\label{sec:conclusions}
This work demonstrates significant promise, as the local MTF has successfully addressed EM scattering by complex objects with multiple homogeneous parameters. These findings underscore the versatility of the current approach, which can efficiently manage diverse and intricate boundary conditions, paving the way for a wide range of new large-scale industrial applications within the EM spectrum. The availability of MTF routines within the open-source bempp-cl framework is expected to drive a surge of experiments within the EM community, fostering new applications and discoveries.

However, certain limitations persist, particularly in terms of computational efficiency and the handling of large-scale systems. To overcome these challenges, future research will focus on implementing the Fast Multipole Method (FMM) \cite{darve2000fast} and adapting its routines to the segments and normal swaps proposed in this work. Additionally, preconditioning strategies will be developed to improve the convergence rates of iterative solvers.

Another key avenue for future exploration involves enhancing the efficient handling of block operators, particularly in the parallelization of matrix-vector products. Furthermore, applications to shape uncertainty quantification \cite{escapil2024shape} in heterogeneous domains will also be investigated.

\section*{Acknowledgement}
The authors thank the support of the following grants: ANID Fondecyt Regular 1231112 and ECOS ANID ECOS230032.

\appendix
\section{\texorpdfstring{Proof of \Cref{prop:calderon_identities}}{Proof of Proposition 1}}
\label{app:appendixA}
For $i\in \{1,\ldots,M\}$, \eqref{eq:cald_temp} ensures the identity as $\bu^\text{sc}_i=\bu_i$. For $i=0$, since the incident field is assumed to satisfy the Maxwell homogeneous equation, one has the exterior Calderón identity in $\Omega_0^c$:
\be \label{eq:cald2}
\left(\frac{\bmI_0}{2}  +   \hbmA_0 \right) \bu_0^\text{inc} = \bu_0^\text{inc} .
\ee
Adding \eqref{eq:cald1} and \eqref{eq:cald2} gives
\be 
 \frac{1}{2}\bu_0^\text{sc} - \hbmA_0 \bu_0^\text{sc}  + \frac{1}{2}\bu_0^\text{inc} + \hbmA_0 \bu_0^\text{inc} = \bu_0 ,
 \ee 
i.e.
\be 
\hbmA_0 ( 2\bu_0^\text{inc} - \bu_0) - \frac{1}{2} \bu_0 = 0.
\ee 
Hence, we retrieve
\be 
 \hbmA_0 (  \bu_0 - 2\bu_0^\text{inc}) + \frac{1}{2} \bu_0 = 0,
\ee 
so that
\be 
 \hbmA_0   \bu_0   +  \frac{1}{2} \bu_0 = 2 \hbmA_0 \bu_0^\text{inc}  = \bu_0^\text{inc}, 
\ee 
the latter giving the desired final result.\qed

\bibliography{references}





\end{document}